\documentclass{daj}

\wlog{Logging}
\makeatletter
\@ifpackageloaded{totpages}{\wlog{Loaded immediately}}{}
\@ifpackageloaded{lastpage}{\wlog{Loaded immediately}}{}
 \usepackage{relsize}

\usepackage{amssymb,amsmath,amsthm,amsfonts,enumitem}
\usepackage{thmtools}

\newcommand{\deffont}[1]{\textnormal{\textsf{#1}}} 

\newcommand{\cC}{\ensuremath{\mathcal{C}}} 
\newcommand{\C}{\ensuremath{\mathbb{C}}}  
\newcommand{\cM}{\ensuremath{\mathcal{M}}}
\newcommand{\cK}{\ensuremath{\mathcal{K}}}
\newcommand{\cN}{\ensuremath{\mathcal{N}}}
\newcommand{\cP}{\ensuremath{\mathcal{P}}}

\newcommand{\cD}{\ensuremath{\mathcal{D}}}
\newcommand{\cI}{\ensuremath{\mathcal{I}}}

\newcommand{\cT}{\ensuremath{\mathcal{T}}}

\newcommand{\cB}{\ensuremath{\mathcal{B}}}

\newcommand{\R}{{\mathbb R}}

\newcommand{\F}{{\mathbb F}}

\newcommand{\N}{\ensuremath{\mathbb{N}}}
\mathchardef\mhyphen="2D

\newcommand{\PR}[1]{\mathrm{Pr}\left[ #1\right]}
\newcommand{\PROver}[2]{\Pr_{#1}\left[ #2\right]}
\newcommand{\E}[1]{\mathbb{E}\left[ #1\right]}

\newcommand{\Eover}[2]{\mathbb{E}_{#1}\left[ #2 \right]}

\newcommand{\inabs}[1]{\left|#1\right|}
\newcommand{\ip}[2]{\ensuremath{\left\langle #1,#2\right\rangle}}

\newcommand{\inset}[1]{\left\{#1\right\}}
\newcommand{\inabset}[1]{\inabs{\inset{#1}}}

\newcommand{\inparen}[1]{\left(#1\right)}

\newcommand{\supp}{\mathrm{supp}}

\newcommand{\poly}{\mathrm{poly}}

\newcommand{\rank}{\ensuremath{\operatorname{rank}}}

\newcommand{\tr}{\mathrm{tr}}

\newcommand{\uniform}{\ensuremath{\mathsf{U}}}
\newcommand{\RS}{\mathrm{RS}}

\newcommand{\weight}{\ensuremath{\operatorname{wt}}}

\newcommand{\DKL}[3]{{D_{\mathrm{KL}}}_{#3}\left(#1\parallel #2\right)}
\newcommand{\CRS}[3][{}]{\RS_{#1}\inparen{#2;#3}}

\newcommand{\ind}[1]{\ensuremath{\mathbf{1}_{#1}}}

\newcommand{\TRLC}{\ensuremath{{\mathrm{RLC}}}}

\newcommand{\PLD}[3]{\ensuremath{{\mathrm{LD}}_{#1,#2,#3}}}

\newcommand{\CRLC}{C_{\mathrm{RLC}}}

\newcommand{\LD}[2]{\ensuremath{\inparen{#1,#2}\text{-list-decodable}}}
\newcommand{\LR}[3]{\ensuremath{\inparen{#1,#2,#3}\text{-list-recoverable}}}

\newcommand{\SC}[2]{\ensuremath{\inparen{#1,#2}\text{-span-clustered}}}
\newcommand{\RSC}[3]{\ensuremath{\inparen{#1,#2,#3}\text{-recovery-span-clustered}}}
\newcommand{\RC}[2]{\ensuremath{\inparen{#1,#2}\text{-recovery-clustered}}}

\newcommand\numberthis{\addtocounter{equation}{1}\tag{\theequation}}

\newcommand{\wh}{\widehat}

\newcommand{\mybigcirc}{\mathlarger{\mathlarger{\mathlarger{{\odot}}}}}

\newcommand{\eps}{\varepsilon}
\renewcommand{\epsilon}{\varepsilon}

\renewcommand{\exp}[1]{\mathrm{exp}\inparen{#1}}

\newcommand{\expOver}[2]{\mathrm{exp}_{#1}\inparen{#2}}

\theoremstyle{plain}
\declaretheorem[name=Theorem,numberwithin=section]{theorem}
\declaretheorem[name=Lemma,sibling=theorem]{lemma}
\newtheorem*{lemma*}{Lemma} 
\newtheorem*{theorem*}{Theorem} 
\newtheorem{definition}[theorem]{Definition}

\newtheorem{observation}[theorem]{Observation}

\newtheorem{corollary}[theorem]{Corollary} 

\newtheorem{remark}[theorem]{Remark}
\newtheorem{claim}[theorem]{Claim}
\newtheorem{proposition}[theorem]{Proposition}

\newtheorem{fact}[theorem]{Fact}

\newtheorem{example}[theorem]{Example}

\declaretheorem[name=Theorem]{ourtheorem}
\newtheorem{ourcorollary}[ourtheorem]{Corollary}
\newtheorem{innercustomthm}{Theorem}
\newenvironment{customthm}[1]
{\renewcommand\theinnercustomthm{#1}\innercustomthm}
{\endinnercustomthm}
\newtheoremstyle{named}{}{}{\itshape}{}{\bfseries}{.}{.5em}{\thmnote{#3}}
\theoremstyle{named}

\newcommand{\diagonals}{\Gamma}

\newcommand{\distr}[1]{{\mathsf{Emp}_{#1}}}

\usepackage{mathtools}
\DeclarePairedDelimiter\ceil{\lceil}{\rceil}

\dajAUTHORdetails{%
  title = {Punctured Low-Bias Codes Behave Like Random Linear Codes}, %% please capitalize all significant words
  author = {Venkatesan Guruswami and Jonathan Mosheiff},
  plaintextauthor = {Venkatesan Guruswami and Jonathan Mosheiff},
}   %%% END \dajAUTHORdetails

\dajEDITORdetails{%
   year={2024},
   number={4},
   received={15 January 2023},   % received date: example: 7 January 2017
   published={7 June 2024},  % published date
   doi={10.19086/da.117574},       % XXX = number of paper, e.g. da006 for paper#6
}   %%% END \dajEDITORdetails

\begin{document}

\begin{frontmatter}[classification=text]
%% EDITOR: this will force the keywords to appear right after the Abstract.
%%   If the abstract is too long and would force the keywords off the
%%   front page, please comment out % [classification=text] above
%%   This way the keywords will be floated on the bottom of the first page
%%   even though the Abstract spills over to the next page.

%%% AUTHOR: Title goes here.  This line is optional.  You must use it
%%   if title has footnote attached or requires nontrivial typesetting,
%%   e.g., inclusion of linebreaks to force nice layout.
%\title{Short Proof of R\"odl's $n^{\log\log n}$ Bound\footnote{This is a footnote to the title}} %% please capitalize all significant words

%%% AUTHOR:
%%% List all authors. If you wish, place grant acknowledgements in \thanks.
%%% In brackets include a short tag for each author.
\author[vg]{Venkatesan Guruswami\thanks{Supported in part by NSF grants
	CCF-1814603 and CCF-1908125, and a Simons Investigator Award. Any opinions, findings, and conclusions or recommendations expressed in this material are those of the author(s) and do not necessarily reflect the views of the National Science Foundation.}}
\author[jm]{{Jonathan Mosheiff}\thanks{Supported in part by an Alon Fellowship, NSF grants
CCF-1814603 and CCF-1908125, and a Simons Investigator Award. Any opinions, findings, and conclusions or recommendations expressed in this material are those of the author(s) and do not necessarily reflect the views of the National Science Foundation.}} 

%%% AUTHOR: Abstract goes here
\begin{abstract}
		Random linear codes are a workhorse in coding theory, and are used to show the existence of codes with the best known or even near-optimal trade-offs in many noise models. However, they have little structure besides linearity, and are not amenable to tractable error-correction algorithms.

\smallskip
In this work, we prove a general derandomization result applicable to random linear codes. 
Namely, in settings where the coding-theoretic property of interest is ``local" (in the sense of forbidding certain bad configurations involving few vectors---code distance and list-decodability being notable examples), one can replace random linear codes (RLCs) with a significantly derandomized variant with essentially no loss in parameters. 
Specifically, instead of randomly sampling coordinates of the (long) Hadamard code (which is an equivalent way to describe RLCs), one can randomly sample coordinates of any code with low bias. Over large alphabets, the low bias requirement can be weakened to just large distance. Furthermore, large distance suffices even with a small alphabet in order to match the current best known bounds for RLC list-decodability. 

\smallskip		In particular, by virtue of our result, all current (and future)
achievability bounds for list-decodability of random linear codes extend \emph{automatically} to random puncturings of any low-bias (or large alphabet) ``mother" code. 
We also show that our punctured codes  emulate the behavior of RLCs on stochastic channels, thus giving a derandomization of RLCs in the context of achieving Shannon capacity as well.
Thus, we have a randomness-efficient way to sample codes achieving capacity in both worst-case and stochastic settings that can further inherit algebraic or other algorithmically useful structural properties of the mother code.
\end{abstract}
\end{frontmatter}

\section{Introduction}\label{sec:intro}

Random linear codes (RLCs) are ubiquitous in coding theory, serving as a fundamental building block in the construction of codes since the works of Shannon. RLCs are extensively studied, and known to enjoy excellent combinatorial properties. In particular, they achieve Shannon capacity, the Gilbert-Varshamov rate vs. distance trade-off, and are list-decodable up to capacity.

An RLC of length $n$ and rate\footnote{More accurately, $R$ is the design rate of the code. The actual rate, $\frac{\rank A}{n}$, equals $R$ if and only if $A$ is of full rank---an event which holds with very high probability. See Lemma \ref{lem:ActualVsDesignRate} for more details.} $0 < R < 1$ over alphabet $\F_q$ is the row span of a matrix $A$, sampled uniformly from $\F_q^{k\times n}$, where $k=Rn$. Here, $\F_q$ denotes the finite field of order $q$. Equivalently, one can think of $A$ as a matrix with $n$ random independent columns, each sampled uniformly from $\F_q^k$.

In this work we study derandomizations of RLCs. Specifically, we consider a code $\cC$ generated by a  matrix whose columns are independently sampled from some distribution $\mu$, where the support of $\mu$ is a much smaller set than $\F_q^k$ (possibly exponentially so). We are able to prove that, under fairly modest and general assumptions on $\mu$, this random code $\cC$ is similar to an RLC with respect to \emph{local properties}, a notion which originates in \cite{MRRSW} and will soon be explained. A special case of this result is that $\cC$ is very likely to achieve \emph{list-decoding capacity}, since RLCs are known to do so (in fact, the convergence of RLCs to list-decoding capacity has been extensively studied). Independently of this result, we also show that, similarly to an RLC, $\cC$ is likely to achieve capacity with regard to every memoryless additive-noise channel.

To describe our result more formally, we turn to the notion of \deffont{punctured} codes. Puncturing is a basic operation by which new codes are constructed from existing ones. A \deffont{puncturing} of a code $\cD\subseteq \F_q^m$ is a code $\cC\subseteq \F_q^n$ (where we usually think of $m$ as being much larger than $n$) whose coordinates are taken from those of $\cD$. More precisely, $\cC$ is a puncturing of $\cD$ if $\cC =\{\inparen{x_{i_1},\dots, x_{i_n}}\mid x=\inparen{x_1\dots x_m} \in \cD\}$, for some integers $i_1,\dots, i_n\in [m]$. We sometimes refer to $\cD$ as the \emph{mother code}. When $i_1,\dots, i_n$ are sampled uniformly and independently from $[m]$, we say that $\cC$ is a \deffont{random $n$-puncturing} of $\cD$. Puncturing generally increases the rate of a code, and we often take $\cD$ to be a code of rate approaching $0$, while the rate of $\cC$ is a constant in $(0,1)$.

An equivalent way to describe an RLC of length $n$ and rate $R$ is as a random puncturing of the \emph{Hadamard code}\footnote{The Hadamard code of length $q^k$ over alphabet $\F_q$ is $\inset{\inparen{\ip xy}_{y\in \F_q^k}\mid x\in \F_q^k}$.} of length $q^k$ (where $k=Rn$). The Hadamard code has very poor rate. Consequently, to obtain a punctured code of length $n$ and constant rate from a Hadamard mother code, the mother code must be taken to have length exponential in $n$. Our motivating question is whether, in this construction, one can replace the Hadamard code by a shorter mother code and still obtain a punctured code with the excellent combinatorial properties of an RLC. In Theorem \ref{thm:IntroLowBias}, below, we give this question a strong positive answer.

Of the many special properties of a Hadamard code, the property crucial for that code's excellent performance as a mother code turns out be its \deffont{optimal bias}. Focusing for simplicity on the case $q=2$, the \deffont{bias} of a binary code $\cD$ of length $m$ is defined as $$\max\inset{\frac{\inabs{\textsf{number of ones in }x - \textsf{number of zeros in }x }}m\mid x\in \cD\setminus\inset{0}}\enspace.$$ The bias of a Hadamard code is the smallest possible, namely, zero. The following informal result  can thus be seen as an extension of the statement ``\textit{A random puncturing of a Hadamard code is an RLC}''.

\begin{customthm}{A}[Main result about puncturing of low-bias codes {(informal version of Theorem \ref{thm:MainLowBiasProperties})}]\label{thm:IntroLowBias}
	Let $\cD$ be a linear\footnote{The linearity of $\cD$ is crucial here. Indeed, there exists a (non-linear) code $\cD$ of small bias such that a random puncturing of $\cD$ is unlikely to be similar, in the sense of Theorem \ref{thm:IntroLowBias}, to either an RLC or to a plain random code. For example, one can take any low-bias code $\cD\subseteq \F_2^m$ that contains two codewords $x$ and $y$, each of which has weight $\frac m2$, such that $x+y$ is the all-ones vector.} code with small bias and let $\cC$ be a random puncturing of $\cD$. Then, $\cC$ is likely to have any \deffont{monotone-decreasing} \deffont{local property} that is typically satisfied by an RLC of similar rate.
\end{customthm}
Theorem \ref{thm:IntroLargeDistance}, below, gives broad conditions under which the hypothesis of Theorem \ref{thm:IntroLowBias} can be further relaxed, from requiring low bias to just large distance of the mother code $\cD$.

Informally, a code property $\cP$ is \deffont{monotone-decreasing} and \deffont{local} if whenever a code $\cC$ does not satisfy $\cP$, there exists a small \emph{``bad set''} of codewords in $\cC$ that bears witness to this fact. By \deffont{monotone-decreasing} we mean that adding codewords to the code can only make the property harder to satisfy. A code that has the same monotone-decreasing local properties as an RLC is said to be \deffont{locally similar} to an RLC. Local properties of codes were originally introduced in \cite{MRRSW}, with the motivation of studying the \deffont{list-decodability} of Gallager's LDPC codes. We turn to explain this connection.

\medskip\noindent\textbf{List-decodability.}
In the model of \deffont{list-decoding}, the goal is to decode beyond the unique-decoding radius. A code is said to be (combinatorially)
\deffont{list-decodable up to radius $\rho$} ($0<\rho < 1-\frac 1q$) if every  Hamming ball of radius $\rho n$ in $\F_q^n$ has intersection of size at most $L$ (where $L$ is small, say constant in $n$) with the code. Being list-decodable is a monotone-decreasing local property. Indeed, to show that a code is \emph{not} list-decodable, it suffices to find a set of $L+1$ codewords that all reside within the same radius $\rho$ Hamming ball. A list-decodable code, when accompanied by a decoding algorithm, will allow the correction of a $\rho$ fraction of errors up to some bounded ambiguity in the worst case. We refer the reader to \cite{Gur-NOW-survey} for a detailed discussion of the motivation, usefulness, and potential of the list-decoding model. 

The list-decodability of RLCs has been extensively studied in previous works \cite{ZyablovP81,GHK,CGV13,Wootters13,RW14RLC,RW18RLC,LiWootters,GLMRSW}. The paper \cite{ZyablovP81} already establishes that RLCs are \deffont{list-decodable up to capacity}. Namely, for any fixed $R$ and $\rho$ satisfying $R < 1-h_q(\rho)$, an \footnote{The $q$-ary entropy function $h_q$ is formally defined in Section \ref{sec:entropy}} RLC of rate $R$ is almost surely list-decodable up to radius $\rho$ with constant list-size $L = L(R,\rho,q)$. The focus of the later works is pinpointing the exact dependence of the list-size $L$ on the other parameters.

Since list-decodability is a monotone-decreasing local property, the aforementioned results about list-decodability of an RLC also apply to any code that is  locally similar to an RLC. Therefore, Theorem \ref{thm:IntroLowBias} yields a powerful reduction, allowing us to apply these results about RLCs to the punctured code $\cC$. Moreover, any positive RLC list-decoding bound discovered in the future would also immediately apply to $\cC$. The latter may be relevant since there are still some gaps in our knowledge of RLC list-decodability, especially in the low-rate non-binary regime. This paradigm originates in \cite{MRRSW}, where it was used to show that Gallager codes tend to have the same local properties as an RLC.

In addition to list-decodability, Theorem \ref{thm:IntroLowBias} also applies to other local properties, such as \deffont{list-recoverability}, and its special case, the \deffont{perfect hashing property}. Hence, Theorem \ref{thm:IntroLowBias} immediately yields a positive result (see Section \ref{sec:ListRecovery}) about the list-recoverability of $\cC$, via reduction to established results about the list-recoverability of an RLC. 

\medskip\noindent\textbf{Theorem \ref{thm:IntroLowBias} and RLC derandomization.}
Perhaps more importantly than its specialization to any concrete local property, Theorem \ref{thm:IntroLowBias} is a statement about the robustness of the mechanism by which an RLC is generated: The theorem says that it is possible to choose a code $\cD$ that is radically different from a Hadamard code, randomly puncture it, and end up with a code that, in a local view, has the same desirable properties as an RLC. 

Theorem \ref{thm:IntroLowBias} allows us great flexibility in choosing the mother code $\cD$. While the only structural property of an RLC is its linearity, the punctured code $\cC$ of Theorem \ref{thm:IntroLowBias} can be made to have additional structure via certain choices of the mother code. For example, $\cD$ can be taken to be a \deffont{dual-BCH} code, namely, a code in which every codeword encodes a low-degree polynomial over $\F_{2^\ell}$ ($\ell \in \N$) by the trace of its evaluations over  $\F_{2^\ell}$. In a random puncturing $\cC$ of $\cD$, the codewords correspond to evaluations over a random subset of $\F_{2^\ell}$. 
It is well-known (see \cite{GR-DualBCH}) that dual-BCH codes have small bias, so Theorem \ref{thm:IntroLowBias} applies. Hence, this code $\cC$ enjoys both an algebraic structure and local similarity to an RLC. 

Enforcing a structure on $\cC$ has potential algorithmic advantages. For example, recall that the local similarity of $\cC$ to an RLC only guarantees with high probability the combinatorial property of list-decodability, but not the existence of an efficient list-decoding algorithm for $\cC$. Indeed, for the RLC ensemble itself, it is very likely that no efficient list-decoding algorithm exists for any constant radius $\rho$. Hopefully, by choosing a suitable structure for $\cC$, one may be able to obtain a code which is not only combinatorially list-decodable up to capacity, but also amenable to efficient list-decoding algorithms.

Another application of Theorem \ref{thm:IntroLowBias} is a direct derandomization of RLCs. Utilizing constructions such as \cite{ABNNR92,Ta-Shma17}, $\cD$ can be taken to be an explicit low-bias linear code of length $O(n)$ (where $n$ is the desired length of $\cC$). With such a short mother code, only $O(n)$ random bits are needed to construct the punctured code $\cC$, which is locally similar to an RLC. This is in contrast to the $O(n^2)$ random bits needed to construct an actual RLC of the same length. While methods to sample linear codes with $O(n)$ randomness were known in some settings, the approach and analysis was tailored to the specific setting (e.g, via random Toeplitz matrices for the Gilbert-Varshamov trade-off). In contrast, our approach applies uniformly for all local properties. For list-decoding, this gives the first linear randomness method to sample codes that achieve the trade-offs of RLCs. \footnote{A new work, published after this paper was accepted for publication, further derandomizes our construction \cite{PP2023}.}

The idea of relating the local properties of a more structured code $\cC$ to those of an RLC already figures prominently in the previously mentioned \cite{MRRSW}, where a \emph{Gallager LDPC Code} is cast in the role of $\cC$. While our methods in the present work are very different, our proof of Theorem \ref{thm:IntroLowBias} does use the framework of \cite{MRRSW}, as well as the \emph{RLC Threshold Theorem} \cite[Thm.\ 2.8]{MRRSW} proven there. 

\medskip\noindent\textbf{Replacing low bias by large distance.}
A linear code of low bias necessarily has large minimal distance. For example, in the binary case, the normalized Hamming weight of a codeword $x\in \cD\setminus \{0\}$ is at least $\frac{1-\mathrm{bias}(\cD)}2$. Theorem \ref{thm:IntroLargeDistance} extends Theorem \ref{thm:IntroLowBias} by considering the case in which the mother code $\cD$ has large minimal distance, but not necessarily low bias. For large enough alphabet size $q$, large minimal distance of $\cD$ is enough to guarantee the conclusion of Theorem \ref{thm:IntroLowBias}. For small $q$, while we do not have such a general result, we are able to use specific characteristics of the the property of list-decodability to prove that $\cC$ is, with high probability, list-decodable up to capacity. 
\begin{customthm}{B}[Main result about puncturing of large-distance codes (informal version of {Theorems \ref{thm:MainLargeDistanceProperties} and \ref{thm:MainReduceToGHK})}]\label{thm:IntroLargeDistance}
	Let $\cD$ be a linear code whose minimal distance is near $1-\frac 1q$, and let $\cC$ be a random puncturing of $\cD$ of rate $R$. Then:
	\begin{enumerate}
		\item Let $\cP$ be a monotone-decreasing local property that is likely satisfied by an $RLC$ of rate $R-2\log_q 2$. Then, $\cC$ is likely to satisfy $\cP$.
		\item Even if its alphabet size is small, $\cC$ is likely to be list-decodable and list-recoverable up to capacity (similarly to an RLC of rate $R$).
	\end{enumerate}
\end{customthm}

\medskip\noindent\textbf{Achieving Shannon capacity with punctured codes.}
Theorems \ref{thm:IntroLowBias} and \ref{thm:IntroLargeDistance} are most relevant when considering the performance of $\cC$ as an error correcting code in the \emph{worst-case error} model. To complete the picture, we also prove the following result, dealing with the \emph{random error} model. It is well-known that RLCs achieve capacity with regard to any memoryless additive noise channel. The following informal theorem generalizes this statement to our punctured code $\cC$.
\begin{customthm}{C}[Puncturing of large-distance codes in stochastic channels {(informal version of Theorem \ref{thm:MainStochastic})}]\label{thm:IntroStochastic}
	Let $\cD$ be a linear code whose minimal distance is near $1-\frac 1q$, and let $\cC$ be a random puncturing of $\cD$ with rate $R$. Let $\cN$ be a memoryless additive noise channel with capacity at least $R+\eps$. Then, it is possible to reliably communicate across $\cN$ using $\cC$.
\end{customthm}

\subsection{Previous work}
Randomly punctured codes have recently gotten a lot of attention, motivated by the study of \deffont{Reed-Solomon (RS) codes} \cite{RW14,ST20, GLSTW21, FKS21,GST21}. The \deffont{RS code of dimension $1\le k\le q$ over the set $S\subseteq \F_q$} is defined as $$\CRS[\F_q]{S}{k} = \inset{\inparen{f(\alpha_1),\dots,f(\alpha_n)}\mid f\in \F_q[x], \ \deg(f) < k}\enspace .$$The length of the code is $n = |S|$. A classical algorithm \cite{GS99} can efficiently list-decode an RS code up to the \deffont{Johnson radius} $1-\sqrt R-o(1)$. 

An important open question is whether efficient list-decoding of some RS codes is possible all the way up to the capacity radius $1-R-o(1)$. A necessary condition for such an algorithmic result is to have RS codes which are \emph{combinatorially} list-decodable up to capacity, and such codes are yet unknown. In fact, even the existence of RS codes that are combinatorially decodable beyond the Johnson bound was only recently proven \cite{ST20}. \footnote{Since this paper was accepted for publication, several new works \cite{BGM2023,GZ2023,AGL2023} proved that randomly punctured RS codes are in fact list-decodable up to capacity.}

Our freedom in constructing an RS code lies mainly in the choice of the evaluation set $S$. A natural choice is to take $S$ to be a uniformly random subset of $\F_q$ of the desired size $n$. When $S$ is sampled this way, the code $\CRS[\F_q]{S}{k}$ is essentially a random puncturing of the \deffont{full RS code} $\CRS[\F_q]{\F_q}{k}$. All of the aforementioned works \cite{RW14,ST20, GLSTW21, FKS21,GST21} take this viewpoint. In particular, in \cite{RW14,FKS21,GST21}, the main results about list-decodability of RS codes are all immediate special cases of more general results about randomly punctured codes. Our Theorems \ref{thm:IntroLowBias} and \ref{thm:IntroLargeDistance} are similar in spirit to the latter. We note that our work is the \emph{first} in this line to yield punctured codes that achieve list-decoding capacity, and we do so \emph{for every choice of rate.} The previous works showed trade-offs that were bounded away from list-decoding capacity for all rates. 

\subsection{On an error in a previous version of this paper regarding RS codes}
The full RS code $\CRS[\F_q]{\F_q}{k}$ for small $k$ has near-optimal distance, and thus it would seem that one could apply Theorem \ref{thm:IntroLargeDistance} to it, proving that an RS code over a random evaluation set is likely to be locally-similar to an RLC, and, in particular, list-decodable up to capacity. Unfortunately, the \emph{large distance} requirement of the theorem becomes stricter as the alphabet $q$ grows, and impossible to achieve whenever $q$ is larger than $n$ (see Theorem \ref{thm:MainLargeDistanceProperties}). An earlier version of the current paper erroneously claimed to overcome this difficulty by passing to trace codes over a smaller alphabet, and use it to deduce the list-decodability of certain RS codes. Later, prompted by a question from Zeyu Guo, we have found an error in the proof and retracted this claim.

\section{Main Results}

Before stating our main results, we formally define some of the relevant notions.

\begin{definition}[Random puncturing]
	Fix some prime power $q$. Let $m,n\in\N$. An ($m\to n$) \deffont{puncturing map} is a function $\varphi:\F_q^m\to \F_q^n$ of the form $\varphi(u=(u_1,\dotsc, u_m)) = (u_{i_1}, u_{i_2}, \cdots , u_{i_n})$ for some $i_1,\dotsc, i_n\in [m]$. If $i_1,\dotsc, i_n$ are sampled i.i.d. and uniformly from $[m]$, we say that $\varphi$ is a \deffont{random ($m\to n$) puncturing map}.
	
	A \deffont{random $n$-puncturing} of a code $\cD\subseteq \F_q^m$ is a random code $\cC = \varphi(\cD) = \{\varphi(u)\mid u\in \cD\}$, where $\varphi:\F_q^m\to \F_q^n$ is a random puncturing map. The \deffont{design rate} of $\cC$ is $\frac {\log_q \inabs{\cD}}n$.
\end{definition}

\begin{definition}
	Let $\cD\subseteq \F_q^m$, where $q$ is a power of some prime $p$, be a linear code and let $\eta > 0$. 		
	\begin{enumerate}
		\item If every $u\in \cD\setminus \{0\}$ has $\weight(u) \ge \frac{(q-1)(1-\eta)}q$, we say that 		$\cD$ has \deffont{$\eta$-optimal distance}. Here, $\weight(u) = \frac{\inabset{i\in [m]\mid u_i \ne 0}}m$ denotes the \deffont{normalized Hamming weight} of $u\in \F_q^m$.
		\item A vector $u\in \F_q^m$ is said to be \deffont{$\eta$-biased} if $\inabs{\sum_{i=1}^m \omega^{\tr(a\cdot u_i)}}\le m\eta$ for all $a\in \F_q^*$.	Here, $\omega = e^{\frac{2\pi i}p}$ and $\tr: \F_q\to \F_p$ is the field trace map (see Section \ref{sec:fourier}). The code $\cD$ is said to be \deffont{$\eta$-biased} if every $u\in \cD\setminus \{0\}$ is $\eta$-biased.			 
	\end{enumerate}
\end{definition}

As shown in Lemma \ref{lem:BiasStrongerThanDistance}, an $\eta$-biased code also has $\eta$-optimal distance, so the former is a stronger notion. For intuition, note that in the binary case $\eta$-bias implies $\frac{1-\eta}2 \le \weight(u) \le \frac{1+\eta}2$ for any $u\in \cD\setminus \{0\}$, whereas $\eta$-optimal distance only implies the lower bound on $\weight(u)$. 

It may be simpler for the reader to focus on the case where $q$ is a prime, i.e., $q=p$. In this case, $\tr$ is merely the identity map. 

If $\cC$ is a random $n$-puncturing of a code $\cD$, the rate of $\cC$ is clearly bounded from above by its design rate. The following lemma shows that when $\cD$ is of almost optimal distance, these two terms are very likely to coincide. In light of this lemma, we blur the distinction between design rate and actual rate.

\begin{lemma}[Actual rate equals design rate whp]\label{lem:ActualVsDesignRate}
	Let $\cD\subseteq \F_q^m$ be a linear code of $\eta$-optimal distance, and let $\cC$ be a length-$n$ random puncturing of $\cD$, of design rate $R \le 1-\log_q(1+\eta q)-\eps$. Then, with probability at least $1-q^{-n\eps}$, the rate of $\cC$ is equal to its design rate.
\end{lemma}
\begin{proof}
	The rate of $\cC$ is smaller than $R$ if and only if there exists a non-zero word $u\in \cD$ such that only coordinates $i\in [m]$ for which $u_i = 0$ are sampled for inclusion in $\cC$. For a given $u$, this happens with probability
	$$(1-\weight(u))^n \le \inparen{\frac 1q + \frac{q-1}q\eta}^n \le \inparen{\frac 1q + \eta}^n = q^{-n(1-\log(1+q\eta))} \ . $$
	The claim follows by a union bound over the non-zero words of $\cD$, of which there are $q^{Rn}-1$, and the assumed upper bound on $R$.		
\end{proof}

\begin{definition}[clustered sets and list-decodability]\label{def:LD}
	Fix $\rho \in [0,1]$. A set of vectors $W\subseteq \F_q^n$ is called \deffont{$\rho$-clustered} if there exists some $z\in \F_q^n$ such that $\weight(u-z)\le \rho$ for each $u\in W$. A code $\cC\subseteq \F_q^n$ is said to be \deffont{$(\rho,L)$-list-decodable} if it does not contain any $\rho$-clustered set of codewords of size $L+1$.
\end{definition}

\subsection{A framework for studying properties of codes}\label{sec:PropertyFramework}
In order to formulate our results, we need to recall some of the framework for studying \emph{local and row-symmetric properties of linear\footnote{This framework makes sense for linear as well as non-linear codes. In this work we restrict ourselves to the linear case.} codes}, established in \cite{MRRSW,GMRSW}.\footnote{The notion of a \emph{local property} from \cite{MRRSW} was later refined and split into two parts in \cite{GMRSW}, where it appears as a \emph{row-symmetric} and \emph{local} property. We follow the latter convention.}

A \deffont{property} $\cP$ of length-$n$ linear codes over $\F_q$ is a collection of linear codes in $\F_q^n$. A linear code $\cC\subseteq \F_q^n$ such that $\cC\in \cP$ is said to \deffont{satisfy} $\cP$. If $\cP$ is upwards closed with regard to containment, it is said to be \deffont{monotone-increasing}.\footnote{While Section \ref{sec:intro} discussed monotone-decreasing properties, it will henceforth be more convenient to deal with monotone-increasing properties. Note that the negation of a monotone-decreasing property is monotone-increasing. Hence, the statement ``the code $\cC$ satisfies every monotone-decreasing local property typically satisfied by an RLC'' is equivalent to ``every monotone-increasing local property typically \emph{not} satisfied by an RLC is also \emph{not} satisfied by $\cC$.}

\begin{definition}[Local and row-symmetric properties]\label{def:LocalRowSymmetric}
	Let $\cP$ be a monotone-increasing property of linear codes in $\F_q^n$. We define the following notions.
	\begin{enumerate}
		\item Fix $b\in \N$. Suppose that there exists a family $\cB_\cP$ of sets of words, such that every $B\in \cB_\cP$ is a subset of $\F_q^n$ with $|B|\le b$, and such that
		$$\cC \text{ satisfies }\cP\quad\iff\quad \exists B\in \cB_\cP~~B\subseteq \cC\enspace.$$
		Then, we say that $\cP$ is a \deffont{$b$-local} property.
		\item Suppose that, whenever a code $\cC\subseteq \F_q^n$ satisfies $\cP$ and $\pi$ is a permutation on $\inset{1,\dotsc,n}$, the code $\inset{\pi x\mid x\in \cC}$ also satisfies $\cP$. We then say that $\cP$ is \deffont{row-symmetric}\footnote{The reason for this terminology will be made clear in Observation \ref{obs:LocalRowSymmetricScalarInvariant}.}. Here, $\pi x$ is the vector obtained by permuting the coordinates of $x$ according to $\pi$.
	\end{enumerate}
\end{definition}

The following is immediate from the definition of $(\rho,L)$-list-decodability. 

\begin{observation}\label{obs:LDGoodProperty}
	Let $q$ be a prime power and $n\in \N$. Fix $\rho \in (0,1)$, $L\in \N$, and let $\cP$ be the monotone-increasing property consisting of codes in $\F_q^n$ that are \textbf{not} $\LD \rho L$. Then, $\cP$ is $(L+1)$-local and row-symmetric.
\end{observation}

Let $\cP$ be a monotone-increasing property over $\F_q^n$. Suppose that $\cP$ is nonempty, namely, that it is satisfied by the complete code $\F_q^n$. We denote its \deffont{threshold} by
$$\TRLC(\cP) = \min\inset{R\in [0,1] \mid \PR{\CRLC^{n,q}(R)\text{ satisfies }\cP}\ge \frac 12}\enspace , $$
where $\CRLC^{n,q}(R)$ is a random linear code of rate $R$ in $\F_q^n$.

This terminology is motivated by the following theorem, which states that the probability of an RLC of rate $R$ satisfying a local, row-symmetric and monotone-increasing property $\cP$, as a function of $R$, rapidly climbs near the threshold from $o(1)$ to $1-o(1)$.

\begin{theorem}[Thresholds for local and row-symmetric properties {\cite[Thm.\ 2.8]{MRRSW}}]\footnote{The theorem as stated in \cite{MRRSW} deals only with the regime of constant $q$ and $b$. The current statement, which allows $q$ and $b$ to depend on $n$, follows by inspecting the proof in \cite{MRRSW}.}\label{thm:RLCSharp}
	Let $\cC\subseteq \F_q^n$ be a random linear code of rate $R$ and let $\cP$ be a monotone-increasing, $b$-local and row-symmetric property over $\F_q^n$, where $\frac{n}{\log_q n} \ge \omega_{n\to\infty}\inparen{q^{2b}}$. The following now holds for every $\eps > 0$.
	\begin{enumerate}
		\item If $R \le \TRLC(\cP)-\eps$ then 
		$$\PR{\cC\text{ satisfies }\cP} \le q^{-(\eps-o_{n\to\infty}(1)) n}\enspace . $$
		\item If $R \ge \TRLC(\cP)+\eps$ then 
		$$\PR{\cC\text{ satisfies }\cP} \ge 1-q^{-(\eps-o_{n\to\infty}(1)) n}\enspace . $$
	\end{enumerate}
\end{theorem}

\subsection{Theorem \ref{thm:IntroLowBias}: Randomly punctured low-bias codes}
Theorem \ref{thm:MainLowBiasProperties}, below, is a formal statement of Theorem \ref{thm:IntroLowBias}.

\begin{restatable}[Puncturings of low-bias linear codes are locally similar to random linear codes]{ourtheorem}{MainLowBiasProperties}\label{thm:MainLowBiasProperties}
	Let $q$ be a prime power, and let $\cP$ be a monotone-increasing, row-symmetric and $b$-local property over $\F_q^n$, where $\frac{n}{\log n}\ge \omega_{n\to\infty}\inparen{q^{2b}}$. Let $\cD\subseteq \F_q^m$ be a linear code. Let $\cC$ be a random $n$-puncturing of $\cD$ of design rate $R \le \TRLC(\cP) - \eps$ for some $\eps > 0$. Suppose that $\cD$ is $\inparen{\frac{\eps b\ln q}{q^b}}$-biased.
	Then,
	$$\PR{\cC\text{ satisfies }\cP} \le q^{-(\eps-o_{n\to\infty}(1)) n}\enspace . $$			
\end{restatable}

\subsection{Applications of Theorem \ref{thm:MainLowBiasProperties} for list-decodability}
In our discussion of list-decoding capacity in Section \ref{sec:intro} we treated the rate $R$ as fixed, and the list-decoding capacity radius $\rho$ as a function of $R$ and the field size. Henceforth we will prefer to think of $R$ as depending on some fixed $\rho$.

The \emph{List-Decoding Capacity Theorem} \cite[Thm.\ 7.4.1]{GRS} states that the \emph{optimal rate} for radius $\rho$ list-decoding over the field $\F_q$ is $R^* = 1-h_q(\rho)$, where 
\begin{equation}\label{eq:h_q}
	h_q(\rho) = -\rho\log_q \rho - (1-\rho)\log_q (1-\rho) + \rho \log_{q} (q-1)
\end{equation} 
is the $q$-ary entropy function (see Section \ref{sec:entropy}). In other words, there exist infinite families of codes of rate $R^*-\eps$ that are list-decodable up to radius $\rho$ with list-size independent of the block length $n$, but no such families exist for rate $R^*+\eps$.

As mentioned in Section \ref{sec:intro}, the list-decodability of RLCs has been extensively studied. Positive and negative results of this sort can be stated, respectively, as lower and upper bounds on $\TRLC\inparen{\PLD \rho Lq}$, where $\PLD \rho Lq$ is the monotone-increasing property of a code over $\F_q$ \emph{not} being $(\rho,L)$-list-decodable. For example, the main result of \cite{GHK} can be stated as
\begin{equation}\label{eq:GHKAsThreshold}
	\TRLC\inparen{\PLD \rho L q} \ge 1-h_q(\rho)-O_{\rho,q}\inparen{\frac 1L}\enspace
\end{equation}
for all $\rho$, $q$ and $L$. In the binary regime, the main result of \cite{LiWootters} together with a negative result from \cite{GLMRSW} yield the very precise bound
\begin{equation}\label{eq:LWAsThreshold}
	1-h_2(\rho)\cdot\frac {L-1}{L-2} \le \TRLC\inparen{\PLD \rho L2} \le 1-h_2(\rho)\frac{L+1}L
\end{equation}
for all $\rho$ and $L\ge 3$. 
By Observation \ref{obs:LDGoodProperty}, $\PLD\rho Lq$ is $(L+1)$-local and row-symmetric, so Theorem \ref{thm:MainLowBiasProperties} applies to it. Plugging in Eqs.\ \eqref{eq:GHKAsThreshold} and \eqref{eq:LWAsThreshold} yields the following corollary.

\begin{ourcorollary}[Puncturings of certain linear codes are locally similar to random linear codes]\label{cor:MainLowBiasLD}
	Fix a prime power $q$, $\rho \in \inparen{0,1-\frac 1q}$, $L\in \N$ and $\eps > 0$. Let $\cD\subseteq \F_q^m$ be an $\inparen{\frac{\eps (L+1)\ln q}{q^{L+1}}}$-biased linear code . Let $\cC$ be a random $n$-puncturing of $\cD$ of design rate $R$, where $\frac n{\log n} \ge \omega_{n\to\infty}\inparen{q^{2(L+1)}}$.
	Then,
	\begin{enumerate}
		\item If $R < 1-h_q(\rho)-\frac{C_{\rho,q}}L - \eps$ then $\cC$ is $\LD \rho L$ with probability $1-q^{-\inparen{\eps-o_{n\to\infty}(1)}n}$. Here, $C_{\rho,q}$ is a constant depending on $\rho$ and $q$.
		\item If $q=2$, $L\ge 3$ and $R < 1-h_2(\rho)\cdot \frac{L-1}{L-2} - \eps$ then $\cC$ is $\LD \rho L$ with probability $1-2^{-\inparen{\eps-o_{n\to\infty}(1)}n}$.
	\end{enumerate}
\end{ourcorollary}
Other positive results about RLC list-decodability (e.g., \cite{Wootters13}) can be similarly used to obtain bounds on the list-decodability of randomly punctured codes.

\deffont{List-recoverability} is another property of interest to which Theorem \ref{thm:MainLowBiasProperties} applies, similarly allowing us to reduce from known results about RLCs. See Section \ref{sec:ListRecovery}	for more details.

\subsection{Derandomization of RLCs}
As discussed in Section \ref{sec:intro}, Theorem \ref{thm:MainLowBiasProperties} can be invoked to derandomize RLCs by casting a code of short block length and low bias in the role of the mother code $\cD$. One result that can be achieved via this method is the following theorem. For simplicity, we focus on the binary case.
\begin{restatable}[Codes locally similar to an RLC with linear randomness]{ourtheorem}{DerandomizationAllProperties}\label{thm:DerandomizationAllProperties}
	There exists a randomized algorithm that, given $b\in \N$, $\eps > 0$, $R^*\in [\eps,1]$ and $n\in \N$, where $\frac{n}{\log_2 n}\ge \omega_{n\to\infty}\inparen{2^{2b}}$ and $n \ge \omega_{n\to\infty}(1/\eps)$, samples a generating matrix for a linear code $\cC\subseteq \F_2^n$ of rate $R = R^*-\eps$ such that
	\begin{equation}\label{eq:DerandomizationALlPropertiesWant1}
		\PR{\cC\text{ satisfies some property }\cP\in \cK} \le 2^{-\Omega(\eps n)}.
	\end{equation}
	Here, $\cK$ is the family of all monotone-increasing, $b$-local and row-symmetric properties $\cP$ over $\F_2^n$ for which the 	threshold $~\TRLC(\cP)$ is at least $R^*$.
	This algorithm uses $O\inparen{n\inparen{b+\log_2\frac 1\eps}}$ random bits, and works in time polynomial in $n$.
\end{restatable}

\subsection{Theorem \ref{thm:IntroLargeDistance}: Randomly punctured codes of near-optimal distance}\label{sec:ResultsGHK}
Theorems \ref{thm:MainLargeDistanceProperties} and \ref{thm:MainReduceToGHK} are detailed versions of, respectively, the two parts of Theorem \ref{thm:IntroLargeDistance}. These theorems extend Theorem \ref{thm:MainLowBiasProperties}, in certain scenarios, to the case where the mother code $\cD$ has near-optimal distance but not necessarily low bias. Theorem \ref{thm:MainLargeDistanceProperties} states that the conclusion of Theorem \ref{thm:MainLowBiasProperties} is still valid, provided that the alphabet $q$ is large enough, and the property $\cP$ is \deffont{scalar-invariant}---a new definition given here.

\begin{definition}[Scalar invariant property]\label{def:ScalarInvariant}
	Let $\cP$ be a property of codes in $\F_q^n$. Suppose that, for every code $\cC$ satisfying $\cP$ and for every diagonal full-rank matrix $\Lambda\in \F_q^{n\times n}$, the code $\Lambda \cC := \inset{\Lambda u\mid u\in \cC}$ also satisfies $\cP$. Then, $\cP$ is said to be \deffont{scalar-invariant}.
\end{definition}
It is not hard to see that $(\rho,L$)-list-decodability (as well as its negation, the property of \emph{not} being $\LD \rho L$) is a scalar-invariant property for all $\rho$ and $L$. The same holds for list-recoverability (see Section \ref{sec:ListRecovery}). 

\begin{restatable}[Puncturings of near-optimal distance linear codes with large alphabet are locally similar to random linear codes]{ourtheorem}{MainLargeDistanceProperties}\label{thm:MainLargeDistanceProperties}
	Let $q$ be a prime power, and let $\cP$ be a monotone-increasing, row-symmetric, $b$-local and scalar-invariant property over $\F_q^n$, where $\frac{n}{\log n}\ge \omega_{n\to\infty}\inparen{q^{2b}}$. Let $\cD\subseteq \F_q^m$ be a linear code of $q^{-b}$-optimal distance. Let $\cC$ be a random $n$-puncturing of $\cD$ of design rate $R \le \TRLC(\cP) - \eps - 2\log_q2$ for some $\eps > 0$. Then,
	$$\PR{\cC\text{ satisfies }\cP} \le q^{-(\eps-o_{n\to\infty}(1)) n}\enspace . $$			
\end{restatable}

Theorem \ref{thm:MainReduceToGHK} deals specifically with list-decodability rather than a general code property. The theorem exploits certain characteristics of the proof of \cite{GHK} (see Eq.\ \eqref{eq:GHKAsThreshold}), and shows that this specific result about list-decodability of RLCs can be applied to our randomly-punctured code $\cC$ as long as the mother code has near-optimal distance, regardless of the alphabet size. The relevant characteristics of \cite{GHK} are discussed in Remark \ref{rem:ReasonForConditions}.

\begin{restatable}[A puncturing of a near-optimal distance code is whp list-decodable up to capacity]{ourtheorem}{mainReduceToGHK}\label{thm:MainReduceToGHK}
	Fix a prime power $q$. Let $L,n\in \N$ and $0 < \rho < \frac{q-1}q$, such that $\frac{n}{\log_q n} \ge \omega_{n\to\infty}\inparen{q^{L+1}}$. Let $\cD\subseteq \F_q^m$ be a linear code with $\eta$-optimal distance, where  $\eta = q^{-L+1}$. Let $\cC$ be a random $n$-puncturing of $\cD$ of design rate $R$, where $R \le 1 - h_q(\rho) - \frac{K}L$ for some constant $K = K_{\rho,q}$. Then, 
	$$\PR{\cC\text{ is }\LD \rho L} \ge 1-q^{-\Omega (n)}\enspace . $$
	
	Furthermore, one can take \begin{equation}\label{eq:CrhoqAsymptotics}
		K_{\rho,q} \le\exp{O\inparen{\frac {(\log q)^2}{\min\inset{(1-1/q-\rho)^2,\rho}}}}
	\end{equation}
	and, in particular, $K_{\rho,q}\le \poly(q)$ whenever $\rho$ is bounded away from $0$ and $1-\frac 1q$. 
\end{restatable}	

\subsection{Theorem \ref{thm:IntroStochastic}: Randomly punctured codes in the stochastic error model}
\begin{definition}[Additive noise channel]
	Let $\nu$ be a distribution over $\F_q$. The \deffont{$\nu$-memoryless additive noise channel} with distribution $\nu$ takes as input a vector $x\in \F_q^n$ and outputs the vector $x+z$, where $z\in \F_q^n$ has entries independently sampled from $\nu$. 
\end{definition}

\noindent For $z\in \F_q^n$, we write $\nu(z) = \prod_{i=1}^n \nu(z_i)$ for the probability of $z$ under the product distribution $\nu^n$.	

The \deffont{capacity} of the $\nu$-memoryless additive noise channel is $1-H_q(\nu)$. Here, $H_q(\nu)$ is the \deffont{base-$q$ entropy}
$$H_q(\nu) = -\sum_{x\in \supp(\nu)}\nu(x)\log_q \nu(x)\enspace. $$

The \deffont{maximum likelihood decoder under uniform prior (MLDU)} for a code $\cC\subseteq \F_q^n$ with regard to the $\nu$-memoryless additive noise channel receives a word $y\in \F_q^n$ and returns a codeword $x\in \cC$ for which $\nu(y-x)$ is maximal. In other words, $x$ maximizes the likelihood of the channel outputting $y$ given input $x$.

It is well-known that an RLC with MLDU decoding achieves capacity with regard to any memoryless additive noise channel. Theorem \ref{thm:IntroStochastic}, stated here formally as Theorem \ref{thm:MainStochastic}, generalizes this fact to random puncturings of low-bias codes.

\begin{restatable}[A puncturing of a low-bias code achieves capacity with regard to memoryless additive noise]{ourtheorem}{MainStochastic}\label{thm:MainStochastic}
	Fix a prime power $q$, a distribution $\nu$ over $\F_q$ and $0 < \eps < 1$. Let $\cD\subseteq \F_q^m$ be an $\frac{\eps}{3(q-1)}$-biased linear code. Let $\cC$ be a random $n$-puncturing of $\cD$ of design rate $R\le 1 - H_q(\nu) - \eps$. Then, with probability $1-q^{-\Omega_\nu\inparen{\eps n}}$, it holds for all $x\in \cC$ that
	\begin{equation}\label{eq:ProbOfCorrectDecoding}
		\PROver{z\sim \nu^n}{\textnormal{the MLDU outputs $x$ on input $x+z$}} \ge 1 - 2q^{-c_\nu \eps^2 n}\enspace,
	\end{equation}
	for some $c_\nu > 0$ that depends only on $\nu$.
\end{restatable} 
We note that the code property of being \emph{decodable with regard to the $\nu$-memoryless additive noise channel}, where $\nu$ is a generic distribution, is not local in the sense of Definition \ref{def:LocalRowSymmetric}. Consequently,
Theorem \ref{thm:MainStochastic} is proven by more direct means, and not as a consequence of the framework described in Section \ref{sec:PropertyFramework}. It is an interesting question whether this framework can be extended to also apply to properties such as decodability and list-decodability in stochastic settings.

\subsection{Organization}
The rest of the paper is organized as follows. In Section \ref{sec:Overview} we demonstrate our techniques by sketching a proof for a weaker version of  Corollary \ref{cor:MainLowBiasLD}. Section \ref{sec:Preliminaries} establishes some general definitions and lemmas used in the main proofs. In Section \ref{sec:PropertySetting} we give more details on the code property framework and discuss the \emph{RLC Threshold Theorem}\cite{MRRSW}. In Section \ref{sec:SimilarityToRLC} we use the above framework to prove Theorems \ref{thm:MainLowBiasProperties} and \ref{thm:MainLargeDistanceProperties} about punctured codes that are locally-similar to an RLC. 
In Section \ref{sec:ProofOfMainReductToGHK} we prove Theorem \ref{thm:MainReduceToGHK} about list-decodability of randomly punctured codes, based on the result of \cite{GHK}. In Section \ref{sec:Derandomization} we prove Theorem \ref{thm:DerandomizationAllProperties}, dealing with derandomization of RLCs. Finally, Theorem \ref{thm:MainStochastic} about punctured codes in the stochastic error model is proved in Section \ref{sec:stochastic}.

\section{Technical warmup}\label{sec:Overview}
For the sake of exposition, we begin by proving Theorem \ref{thm:WeakPuncturing}---a weaker version of Corollary \ref{cor:MainLowBiasLD}.  Theorem \ref{thm:WeakPuncturing} showcases the techniques by which we prove our main results in a simplified setting. Rather than reducing from state of the art results about RLC list-decodability, Theorem \ref{thm:WeakPuncturing} is proven directly, resulting in worse bounds on the list-size. For simplicity, we restrict ourselves to the binary regime. 

\begin{ourtheorem}\label{thm:WeakPuncturing}
	Let $\rho \in \inparen{0,\frac 12}$ and $L\in \N$. Then, there exist $\eta(L) > 0$ and $\eps(L) > 0$ with $\eps(L)\xrightarrow[L\to \infty]{} 0$, such that the following holds.
	Let $\cD\subseteq \F_2^m$ be a linear $\eta$-biased code, and let $\cC$ be a random $n$-puncturing of $\cD$ of design rate $R \le 1-h_2(\rho)-\eps$. Then $\cC$ is $\LD \rho L$ with high probability as $n\to \infty$.
\end{ourtheorem}
\begin{proof}[Proof sketch]
	Let $\varphi : \F_2^m\to \F_2^n$ be the random puncturing map by which $\cC$ is generated from $\cD$. 
	Write $b = \lceil \log_2 (L+1) \rceil$. Now any set of $L+1$ vectors in $\F_2^n$ must contain a subset of $b$ linearly-independent vectors. In particular, for $\cC$ to contain a $\rho$-clustered set of size $L+1$, it must contain a $\rho$-clustered set of $b$ linearly-independent vectors (this argument originated in \cite{ZyablovP81}). Thus, the probability, taken over the random puncturing $\varphi$, that $\cC\text{ is }\textbf{not} \LD\rho L$ is at most
	\begin{align*}	
		& \PR{\exists v_1,\dotsc, v_b\in \cC \text{ which are }\rho\text{-clustered and linearly-independent}} \\
		& = \PR{\exists u_1,\dotsc, u_b\in \cD \text{ s.t.\ }\varphi(u_1),\dotsc,\varphi(u_b) \text{ are }\rho\text{-clustered and linearly-independent}} \\ 
		&\le \PR{\exists u_1,\dotsc, u_b\in \cD \text{ which are linearly independent, s.t.\ }\varphi(u_1),\dotsc,\varphi(u_b) \text{ are }\rho\text{-clustered}}\\ 
		& \le \sum_{\substack{u_1,\dotsc u_b\in \cD\\\text{linearly independent}}}\PR{\varphi(u_1),\dotsc, \varphi(u_b)\text{ are $\rho$-clustered}}\enspace,
	\end{align*}
	where the penultimate inequality is because linear-independence of $u_1,\dotsc, u_b$ is a necessary condition for linear-independence of $\varphi(u_1),\dotsc, \varphi(u_b)$. The sum in the final line has at most $\inabs{\cD}^{b} = 2^{bRn}$ terms, so it suffices to show that 
	\begin{equation}\label{eq:WeakPuncturingTheorem1}
		\PR{\varphi(u_1),\dotsc, \varphi(u_b)\text{ are $\rho$-clustered}} \le 2^{-bRn-\omega(1)}
	\end{equation}
	whenever $u_1,\dotsc, u_b\in \cD$ are linearly independent.
	
	Let $B\in \F_2^{m\times b}$ be the matrix whose columns are $u_1,\dotsc, u_b$, and let $\sigma$ denote the distribution, over $\F_2^b$, of a uniformly random row of $B$. Let $A\in \F_2^{n\times b}$ be the matrix whose columns are $\varphi(u_1),\dotsc, \varphi(u_b)$. A crucial observation is that $A$ is a random matrix whose rows are sampled independently from $\sigma$. At this point, if $\sigma$ were the uniform distribution over $\F_2^b$, we would be done. Indeed, $\sigma$ being uniform means that the columns of $A$, call them $c_1, c_2, \dots,c_b$, are sampled independently and uniformly from $\F_2^n$. This establishes  Eq.\ \eqref{eq:WeakPuncturingTheorem1} since
	\begin{align}
		\PROver{c_1,\dotsc, c_b\sim \uniform(\F_2^n)}{c_1,\dotsc, c_b\text{ are $\rho$-clustered}} &\le \sum_{z\in \F_2^n} \sum_{y_1,\dots,y_b \in B(z,\rho n)} \PROver{c_1,\dotsc, c_b\sim \uniform(\F_2^n)}{\bigwedge_{i=1}^b (c_i = y_i)}  \nonumber  \\
		& = \sum_{z\in \F_2^n} \sum_{y_1,\dots,y_b \in B(z,\rho n)} (2^{-b})^n \label{eq:uniform-to-eta-switch} \\ 
		&\le  \sum_{z\in \F_2^n} 2^{b h_2(\rho) n} (2^{-b})^n  =2^{n\inparen{bh_2(\rho)-b+1}} \nonumber \\ 
		& \le 2^{-bRn-n}\enspace, \label{eq:R-using-eps} 
	\end{align}
	where $B(z,\rho n)$ denotes the Hamming ball of radius $\rho n$ around $z$, and the last inequality Eq.\ \eqref{eq:R-using-eps} holds for, say, $\eps = \frac 2b$. Note that $\eps \le O\inparen{\frac 1{\log L}}$.
	
	We now use a certain formulation of the \emph{Vazirani XOR-Lemma} (see, e.g., \cite{3XOR}) to show that $\sigma$ is in fact arbitrarily close to the uniform distribution over $\F_2^b$. This allows us to finish the theorem by extending the above argument from uniform $\sigma$ to almost-uniform $\sigma$.  
	\begin{lemma}[Vazirani XOR-Lemma]\label{lem:VaziraniOriginal}
		Let $\sigma$ be a distribution over $\F_2^b$ such that for every $y\in \F_q^b\setminus \{0\}$, we have 
		$\tfrac {1-\eta} 2\le \PROver{x\sim \sigma}{\ip xy = 1} \le \tfrac{1+\eta}2$. Then, $\sigma$ is $\inparen{2^b\cdot \eta}$-close in \emph{total-variation distance} to the uniform distribution over $\F_2^b$.
	\end{lemma}
	In our case, $\PROver{x\sim \sigma}{\ip xy = 1} = \weight(By)$. Since the columns of $B$ belong to $\cD$ and are linearly-independent, $By$ is a non-zero codeword of $\cD$. Our assumption about $\cD$ having small bias means that $\weight(By)$ is very close to $\frac 12$, so the hypothesis of Lemma \ref{lem:VaziraniOriginal} is satisfied. Thus, in the above calculation the rows of $A$ are sampled i.i.d from a distribution $\sigma \sim \F_2^b$ which has statistical distance at most $2^b \eta$ from uniform. Therefore, we can replace the $2^{-b}$ term  in Eq.\ \eqref{eq:uniform-to-eta-switch} by an upper bound $(2^{-b} + 2^b \eta)$. By taking $\eta$ small enough, say at most $2^{-2b}$, the bound in Eq.\ \eqref{eq:R-using-eps} remains valid by slightly adjusting parameters (e.g., taking $\eps = \tfrac{3}{b}$). 
\end{proof}

\section{Preliminaries}\label{sec:Preliminaries}
\subsection{General notation}
We denote the uniform distribution over a finite nonempty set $S$ by $\uniform(S)$. 

For $a,b\in \R$, we denote $\expOver ab = a^b$.

The constants implied by asymptotic notation are universal unless stated otherwise. To indicate that the hidden constant may depend on, e.g., for the parameter $p$, we write ``$O_p(\cdot)$''.

If $A\in \F_q^{m\times b}$ and $\cC\subseteq \F_q^m$, we write $A\subseteq \cC$ to mean that each column of $A$ is a codeword in $\cC$. Given a puncturing map $\varphi:\F_q^m\to \F_q^n$, let $\varphi(A)$ denote the matrix obtained from $A$ by applying $\varphi$ to each column. 

\subsection{A characterization of list-decodable linear codes}
Recall the notion of a $\rho$-clustered set (Definition \ref{def:LD}.)
\begin{definition}
	Fix $\rho \in [0,1]$, $L\in \N$. A matrix $A\in \F_q^{n\times b}$ ($1\le b\le L$) with $\rank A = b$ is \emph{$(\rho,L)$-span-clustered} if the column-span of $A$ contains a $\rho$-clustered set of size $L$. 
\end{definition}
Note that for a linear code $\cC \subseteq \F_q^n$ we have
\begin{equation}\label{eq:SpanClusteredCharacterization}
	\cC\text{ is \emph{not} }\LD \rho L \iff \exists A\in \F_q^{n\times b} \text{such that $A$ is $(\rho,L+1)$-span-clustered and }A\subseteq \cC\enspace .
\end{equation} 
Furthermore, we can always take $b$ to be in the range $[\log_q (L+1),L+1]$. Indeed, a matrix of rank smaller than $\log_q (L+1)$ cannot be $(\rho,(L+1))$-span-clustered since its span has cardinality smaller than $L+1$. On the other hand, a rank larger than $L+1$ is never needed since, given a $\rho$-clustered set $W\subseteq \cC$ with $|W|=L+1$, one can take $A$ to be a matrix whose columns are a maximal linearly independent subset of $W$.	

\subsection{The scalar-multiplied code $\Lambda \cC$ and scalar-expanded code $\cD^*$}\label{sec:ScalarExpandedCode}
Let 
$$\diagonals_n = \inset{\Lambda \in \F_q^{n\times n}\mid \Lambda\text{ is diagonal and of full-rank}}\enspace . $$
Recall that, for $\Lambda \in \diagonals_n$ and a code $\cC\subseteq \F_q^n$, we denote $\Lambda\cC = \inset{\Lambda u\mid u\in \cC}$ (Definition \ref{def:ScalarInvariant}). If $\cP$ is a scalar-invariant property (such as list-decodability), the question of whether a code $\cC$ satisfies $\cP$  reduces to that of any code of the form $\Lambda \cC$. To take advantage of this reduction, we shall focus on the random code $\Lambda \cC$ where $\Lambda\sim \uniform(\diagonals_n)$. If $\cC$ is a random puncturing of some code $\cD$, we can realize the code $\Lambda\cC$ as a puncturing of the code $\cD^*$, which we now define.	

\begin{definition}\label{def:ScalarExpansion}
	Given $u\in \F_q^m$, let $u^*\in \F_q^{m(q-1)}$ denote the vector $$u^* = \mybigcirc_{a\in \F_q^*} \ (au)\enspace , $$
	where $\mybigcirc$ stands for concatenation of vectors. Given a matrix $B\in \F_q^{m\times b}$ with columns $a_1,\dotsc, a_b$, let $B^*\in \F_q^{m(q-1)\times b}$ be the matrix whose columns are $a_1^*,\dotsc, a_b^*$. 
	Denote $\cD^* = \{u^*\mid u\in \cD\}\subseteq \F_q^{m(q-1)}$. 
\end{definition}

\begin{observation}\label{obs:LambdaCAsAPuncturing}
	The code $\Lambda \cC$, where $\Lambda \sim \uniform(\diagonals_n)$ and $\cC$ is a random $n$-puncturing of $\cD$, is distributed identically to a random $n$-puncturing of $\cD^*$.
\end{observation}	

\subsection{Fourier transform}\label{sec:fourier}
We recall the following elementary facts about the Fourier transform\footnote{Our convention is to use counting norm for $f$ and expectation norm for $\hat f$.} of a function $f: \F_q^b\to \C$.
\begin{definition}[Fourier (and inverse Fourier) transform]\label{def:fourier}
	Suppose that $q = p^r$ for some prime $p$, and let $\omega = e^{\frac{2\pi i}p}$. Let $b\in \N$ and let $f: \F_q^b\to \C$. Then $\wh f: \F_q^b \to \C$ is defined by
	$$\wh f(y) = \sum_{x\in \F_q^b} f(x) \cdot \overline{\chi_y(x)}, \quad \text{where} \quad
	\chi_y(x) = \omega^{\tr\inparen{\ip xy}}\enspace . $$	
	Here, $\tr:\F_q\to \F_p$ stands for the field trace function $\tr(x) = \sum_{i=0}^{r-1} x^{p^i}$. We also have the Fourier inversion formula:
	$$f(x)  = q^{-b} \sum_{y \in \F_q^b} \wh f(y) \chi_y(x)  \enspace. $$
\end{definition}

\begin{fact}[Parseval's identity]\label{fact:Parseval}
	Let $f,g:\F_q^b\to \C$. Then, 
	$$\sum_{x\in \F_q^b} f(x)\overline{g(x)} = \Eover{y\sim \uniform(\F_q^b)}{ \wh f(y)\overline{\wh g(y)}}\enspace . $$
	In particular,
	$\sum_{x\in \F_q^b} \inabs{f(x)}^2 = \Eover{y\sim \uniform(\F_q^b)}{ \inabs{\wh f(y)}^2}$.
\end{fact}

\subsection{Entropy and KL-divergence}\label{sec:entropy}
Given $x\in [0,1]$, we write
$$h_q(x) = -x\log_q x - (1-x)\log_q (1-x) + x \log_{q} (q-1)$$
for the base-$q$ entropy of a random variable over $\{0,\dotsc, q-1\}$, which takes $0$ with probability $1-x$ and each $i\in \{1,\dotsc, q-1\}$ with probability $\frac{x}{q-1}$. 

The \deffont{$q$-ary Kullback-Leibler Divergence} of two distributions $\tau,\sigma$ over a finite set $S$ is
$$\DKL \tau \sigma q = \sum_{s\in S}\tau(s) \log_q\frac{\tau (s)}{\sigma(s)}\enspace . $$

\subsection{The empirical distribution of the rows of a matrix}

\begin{definition}
	Given a vector $a\in \F_q^n$ we define its \deffont{empirical distribution} $\distr a$ over $\F_q$ by
	$$\distr a(x) = \PROver{i\in [n]} {a_i = x}\enspace . $$
	More generally, given $A\in \F_q^{n\times b}$, let $\distr A$ denote its \emph{empirical row distribution}, that is, the distribution over $\F_q^b$ defined by
	$$\distr A(x) = \PROver{i\in [n]} {A_i = x}\enspace , $$
	where $A_i$ denotes the $i$'th row of $A$.
\end{definition}
The distribution $\distr A$ characterizes the matrix $A$ up to row permutations. Indeed, it is immediate that $A,A'\in \F_q^{n\times b}$ satisfy $\distr A = \distr{A'}$ if and only if $A'$ is obtained from $A$ by a row permutation.

\begin{fact}[{\cite[Thm.\ 11.1.4]{CoverThomas}}]\label{fact:DKL}
	Let $A\in \F_q^{n\times b}$ have rows sampled identically and independently from some distribution $\sigma$ over $\F_q^b$. Then, for any distribution $\tau$ over $\F_q^b$, 	
	$$\PR{\distr A = \tau} \le q^{-\DKL \tau\sigma q\cdot n}\enspace . $$
\end{fact}

\begin{definition}
	Let $\tau$ be a distribution over $\F_q^b$. We denote $\dim(\tau) = \dim \supp(\tau)$. If  $\dim(\tau)=b$, we say that $\tau$ is a \deffont{full-rank} distribution. 
\end{definition}

\begin{definition}[Matrix of a particular distribution]
	Let $\tau$ be a distribution over $\F_q^b$ (where $b\in \N$). For $n\in \N$, we denote $$\cM_{n,\tau} = \inset{A\in \F_q^{n\times b}\mid \distr A = \tau}\enspace . $$
\end{definition}

A distribution $\tau$ over $\F_q^b$ is said to be \deffont{$n$-feasible} if $\tau(x)\cdot n$ is an integer for all $x\in \F_q^b$.  Observe that any $n$-feasible distribution over $\F_q^b$ corresponds to a partition of $n$ identical balls into $q^b$ buckets. The bound below thus follows immediately.
\begin{fact}\label{fact:NumberOfTypes}
	The number of $n$-feasible distributions over $\F_q^b$ is at most $(n+1)^{q^b}$.
\end{fact}
Clearly, $n$-feasibility of $\tau$ is a necessary condition for $\cM_{n,\tau}$ to be nonempty. When this condition holds, $|\cM_{n,\tau}|$ is equal to the multinomial coefficient $\frac{n!}{\prod_{x\in\F_q^b}(\tau(x)n)!}$. By standard bounds on multinomial coefficients, we have 
\begin{equation}\label{eq:M_nTauApproximation}
	n^{-O(q^b)}\cdot q^{n\cdot H_q(\tau)}\le \inabs{\cM_{n,\tau}} \le q^{n\cdot H_q(\tau)}\enspace.
\end{equation}

\subsubsection{The Fourier transform of an empirical distribution}
We record several useful properties of the function $\wh{\distr A}$ for a given matrix $A$. The following is immediate.
\begin{fact}\label{fact:BiasAndFourier}
	A vector $u\in \F_q^n$ is $\eta$-biased ($\eta > 0$) if and only if $\inabs{\wh{\distr u}(a)} \le \eta$ for all $a\in \F_q^*$.
\end{fact}

The following identity shows that the Fourier transform of $\distr A$ (where $A\in \F_q^{n\times b}$) is in fact composed of the Fourier transforms of $\distr{Ay}$ over $y\in \F_q^b$. Let $a\in \F_q$. Then,
\begin{equation}\label{eq:FourierMatrixImage}
	\wh{\distr A}(ay) = \sum_{x\in \F_q^b}\distr A(x) \omega^{-\tr (a\ip xy)} = \Eover{x\sim \distr A} {\omega^{-\tr(a\ip xy)}} = \Eover{z\sim \distr{Ay}}{\omega^{-\tr(az)}} = \wh{\distr{Ay}}(a) \enspace.
\end{equation}

By Fact \ref{fact:Parseval}, the normalized Hamming Weight of a vector $u\in \F_q^n$ can be conveniently expressed in terms of the Fourier transform of $\distr u$. 
\begin{equation}\label{eq:WeightFourier}
	\weight(u) = \sum_{x\in \F_q} \ind{x\ne 0} \cdot \distr u(x) = \frac{q-1}q\cdot  \wh{\distr u}(0) - \frac{1}q\cdot  \sum_{a\in \F_q^*}\wh{\distr u}(a) = \frac{q-1}q - \frac{1}q\cdot  \sum_{a\in \F_q^*}\wh{\distr u}(a).
\end{equation}
This yields the following relation between bias and weight.
\begin{lemma}\label{lem:VectorBiasAndWeight}
	Let $u\in \F_q^n$ be $\eta$-biased for some $\eta > 0$. Then $$\frac{q-1}q(1-\eta)\le\weight(u) \le \frac{q-1}q(1+\eta)\enspace.$$
\end{lemma}	
\begin{proof}
	By Eq.\ \eqref{eq:WeightFourier} and Fact \ref{fact:BiasAndFourier},
	$$
	\inabs{\weight(u)-\frac{q-1}q} = \inabs{\frac 1q\cdot \sum_{a\in \F_q^*}\wh{\distr u}(a)}  \le \frac{q-1}q\cdot \eta\enspace.\qedhere
	$$
\end{proof}
We have the following immediate conclusion.
\begin{lemma}\label{lem:BiasStrongerThanDistance}
	For any $\eta \ge 0$, an $\eta$-biased code also has $\eta$-optimal distance.
\end{lemma}	

\section{Thresholds for monotone, local and row-symmetric properties}\label{sec:PropertySetting}
In this section we reformulate several notions from previous works that relate to code properties \cite{MRRSW,GLMRSW,GMRSW}. Our goal is to eventually restate the \emph{RLC Threshold Theorem} of \cite{MRRSW} (see Theorem \ref{thm:RLCThreshold}) and show how it can be used to prove that a given random code ensemble is locally-similar to an RLC. This technique will be applied in Section \ref{sec:SimilarityToRLC} (via Lemma \ref{lem:ReductionToRLC}) to prove Theorems \ref{thm:MainLowBiasProperties} and \ref{thm:MainLargeDistanceProperties}.

\subsection{Monotone-increasing properties and minimal sets}
A monotone-increasing property $\cP$ of codes has a unique \deffont{minimal-set} $\cM_{\cP}$, namely, a matrix $A\subseteq \F_q^{n\times b}$ ($b\in \N$) of full column-rank belongs to $\cM_{\cP}$ if the code consisting of the column-span of $A$ satisfies $\cP$, but no proper linear subspace of that code does so.

\begin{example}[Minimal set for list-decodability]\label{ex:LDAsProperty} Fix a prime power $q$, $n,L\in \N$ and $\rho \in [0,1]$. Consider the monotone-increasing property $\cP$ consisting of all linear codes in $\F_q^n$ that are \textbf{not} $\LD \rho L$. By Eq.\ \eqref{eq:SpanClusteredCharacterization}, 
	$$\cM_{\cP} \subseteq \inset{A\subseteq \F_q^{n\times b}\mid b\in\N,~~~A \text{ is }\SC{\rho}{L+1}}\enspace .$$
\end{example}	

We can reformulate the notions of \deffont{local}, \deffont{row-symmetric} and \deffont{scalar-invariant} monotone-increasing properties in terms of the minimal set $\cM_\cP$ (see Definitions \ref{def:LocalRowSymmetric} and \ref{def:ScalarInvariant}).

\begin{observation}[Local, row-symmetric and scalar-invariant properties in terms of $\cM_{\cP}$]\label{obs:LocalRowSymmetricScalarInvariant}
	Let $\cP$ be a monotone-increasing property of linear codes in $\F_q^n$. 
	\begin{enumerate}
		\item  Fix $b\in \N$. Then, $\cP$ is \deffont{$b$-local} if and only every matrix in $\cM_{\cP}$ has at most $b$ columns.
		\item The property $\cP$ is \deffont{row-symmetric} if and only if, for each $A\in \cM_{\cP}$, it holds that every matrix obtained by permuting the rows of $A$ also belongs to $\cM_{\cP}$.
		\item The property $\cP$ is \deffont{scalar-invariant} if, for each $A\in \cM_{\cP}$ and every full-rank diagonal matrix $\Lambda \in \F_q^{n\times n}$ it holds that $\Lambda A \in \cM_{\cP}$.
	\end{enumerate}
\end{observation}

\subsection{Row-symmetric $b$-local properties in terms of distributions over $\F_q^b$}\label{sec:LocalPropertiesAndDistributions}
Thresholds for row-symmetric and local properties can be characterized in terms of empirical distributions of certain matrices. We recall this connection, first proven in \cite{MRRSW}.

\begin{definition}[Minimal types]\label{def:MinimalTypes}
	Let $\cP$ be a monotone-increasing, $b$-local, row-symmetric property over $\F_q^n$. We denote $\cT_\cP:= \{\distr M \mid M\in \cM_{\cP}\}$. 
\end{definition}
The row-symmetry of $\cP$ implies that $\cM_\cP$ is closed to row permutations. Hence, the membership of a matrix $M$ in $\cM_\cP$ depends only on $\distr M$. Therefore,
\begin{equation}\label{eq:MinimalTypesToSets}
	\cM_\cP=\bigcup_{\tau \in \cT_\cP} \cM_{n,\tau}\enspace.
\end{equation}

Working with the set $\cT_\cP$ rather than $\cM_\cP$ is useful because the former is of size polynomial in $n$, while the latter's size is typically exponential in $n$. Indeed, Fact \ref{fact:NumberOfTypes} immeidately implies the following.
\begin{fact}\label{fact:PropertyDecomposition}
	In the setting of Definition \ref{def:MinimalTypes}, $|\cT_\cP|\le (n+1)^{q^b}$.
\end{fact}

We demonstrate the notions of $\cM_\cP$ and $\cT_\cP$ via the property of $(\rho,L)$-list-decodability.

\begin{definition}\label{def:SpanClusteredType}
	Fix a prime power $q$. Let $b,n\in N$ and let $\tau$ be an $n$-feasible distribution over $\F_q^b$. If a matrix $A\in \cM_{n,\tau}$ is $(\rho,L+1)$-span-clustered, we say that $\tau$ is \deffont{$(\rho,L+1)$-span-clustered (with regard to $n$)}.
\end{definition}
\begin{remark}
	Note that $\tau$ is well-defined. Indeed, if a matrix $M$ is $\SC{\rho}{L+1}$ then the same holds for every row-permutation of $M$.
\end{remark}

\begin{example}[Minimal types for list-decodability]
	Let $\cP$ be the property consisting of all linear codes in $\F_q^n$ that are \textbf{not} $\LD \rho L$. By Example \ref{ex:LDAsProperty} and Eq.\ \eqref{eq:MinimalTypesToSets},
	$$\cT_{\cP} \subseteq \inset{\tau \mid \tau \text{ is a $\SC \rho{L+1}$ $n$-feasible distribution over $\F_q^{b}, \: b\le L+1$}\enspace}.$$
\end{example}

It is revealing (although not necessary) to develop a direct characterization of the set of $\SC \rho{L+1}$ distributions. We do so in Appendix \ref{appendix:SpanClusteredTypes}.

To state the Threshold Theorem from \cite{MRRSW} we also require the notion of an \deffont{implied distribution}.

\begin{definition}[Implied distribution {\cite[Def.\ 2.6]{MRRSW}}]\label{def:ImpliedDistribution}
	Let $\tau$ be a distribution over $\F_q^b$ and let $D\in \F_q^{b\times a}$ such that $\rank D = a$ for some $a\le b$. The distribution (over $\F_q^a$) of the random vector $xD$, where $x\sim \tau$ (note that $x$ is a row vector), is said to be \deffont{$\tau$-implied}. We denote the set of $\tau$-implied distributions by $\cI_\tau$. 
\end{definition}
The motivation for Definition \ref{def:ImpliedDistribution} is the following observation, which follows immediately from the linearity of the code.

\begin{observation}\label{obs:ImpliedMotivation}
	Let $\tau$ be a distribution over $\F_q^b$, and let $\tau'\in \cI_{\tau}$. Then, any linear code containing a matrix in $\cM_{n,\tau}$ must also contain some matrix in $\cM_{n,\tau'}$. 
\end{observation}

We now have the following characterization of the threshold. 

\begin{theorem}[{\cite[Thm.\ 2.8]{MRRSW}}]\footnote{The precise error term does not appear in the statement of this theorem in \cite{MRRSW}, but follows by inspecting the proof there.}\label{thm:RLCThreshold}
	Let $\cP$ be a monotone-increasing, $b$-local, row-symmetric property over $\F_q^n$, and let $\cT_\cP$ be as in Fact \ref{fact:PropertyDecomposition}. Then, $$\TRLC(\cP) = \min_{\tau \in \cT_\cP} \max_{\tau' \in \cI_\tau} \inparen{1-\frac{H_q(\tau')}{\dim(\tau')}} \pm \frac{2q^{2b}\log_q n}n\enspace . $$ 
\end{theorem}
Below, we demonstrate Theorem \ref{thm:RLCThreshold} via the property of \deffont{list-recoverability}. The motivation is two-fold. First, list-recoverability itself is a property of significant interest. Secondly, we will use Theorem \ref{thm:RLCThreshold} in the proof of Theorems \ref{thm:MainLowBiasProperties} and \ref{thm:MainLargeDistanceProperties}, and the special case of list-recoverability will help familiarize the reader with this tool.

\subsection{List-recoverability as a property of codes}\label{sec:ListRecovery}
\deffont{List-recovery} is an important generalization of list-decoding, where the decoder is given not one, but a subset of $\ell$ symbols per position, and the goal is to list all codewords which ``miss" at most a fraction $\rho$ of these subsets. We formally define the notion of (combinatorial) list-recoverability.

\begin{definition}	\label{def:LR}
	Fix $1\le \ell \le q$ and let $\rho\in (0,1-\ell/q)$. The set $W$ is said to be \deffont{$(\rho,\ell)$-recovery-clustered} if there exist sets $Z_1,\dots,Z_n \subseteq \F_q$, each of which is of size at most $\ell$, such that $\inabset{i\in [n] \mid u_i\notin Z_i} \le \rho n$ for all $u\in W$. A code $\cC\subseteq \F_q^n$ is called \deffont{$\LR \rho \ell L$} if it does not contain any $(\rho,\ell)$-recovery-clustered set of size $L+1$.
\end{definition}

The following is immediate:
\begin{observation}\label{obs:ListRecoveryGoodProperty}
	Fix a prime power $q$, $1\le \ell < q$,  $\rho\in (0,1-\ell/q)$ and $L\in \N$. Let $\cP$ denote the monotone-increasing linear-code property of codes in $\F_q^n$ that are \emph{not} $\LR \rho \ell L$. Then, $\cP$ is row-symmetric, $(L+1)$-local and scalar-invariant.
	
	Say that a matrix $A\in \F_q^{n\times b}$ is \deffont{$\RSC \rho\ell L$} if the column span of $A$ contains a $(\rho,\ell)$-clustered set of size $L+1$. Then, $\cM_\cP$ is contained in the set of all $\RSC \rho \ell L$ matrices.
\end{observation}

Note that list-recovery generalizes list-decodability (Definition \ref{def:LD}), i.e., a set $W$ is $\rho$-clustered if and only if it is $(\rho,1)$-recovery-clustered. Likewise, a code is $\LD \rho L$ if and only if it is $\LR \rho 1 L$. 

The \emph{List-Recovery Capacity Theorem} \cite[Thm.\ 2.4.12]{NicThesis} gives the threshold rate for list-recoverability as $R^* = 1-h_{q,\ell}(\rho)$, where $h_{q,\ell}(\rho) = \rho\log_q \inparen{\frac {q-\ell}\rho} + (1-\rho)\log_q \inparen{\frac \ell{1-\rho}}$. Namely, for every $\eps > 0$ there exists a family of $\LR\rho \ell {O_{\rho,\ell,\eps}(1)}$ codes of rate at least $R^*-\eps$ but, on the other hand, every $\LR\rho \ell L$ family of codes of rate $\ge R^*+\eps$ has $L$ exponentially large in $\eps n$. 

We now use Theorem \ref{thm:RLCThreshold} to prove that RLCs achieve list-recovery capacity\footnote{Better lower bounds on $\TRLC^{n,q}(\rho,\ell,L)$ are known. See, e.g., \cite{RW18RLC}.}.

\begin{proposition}[RLCs achieve list-recovery capacity]\label{prop:RLCListRecoveryCapacity}
	For any fixed $q,\rho, \ell$ and $L$, we have
	$$\TRLC^{n,q}(\rho,\ell,L) \ge 1-h_{q,\ell}(\rho) - \frac \ell {\log_q L} - o_{n\to \infty}(1)\enspace.$$
\end{proposition}
Proposition \ref{prop:RLCListRecoveryCapacity} means that RLCs get to within $\eps$ of the capacity rate for list-recovery with list-size $L \approx q^{\frac \ell\eps}$.

To prove Proposition \ref{prop:RLCListRecoveryCapacity} we need the following claim.
\begin{claim}\label{claim:HBoundForRCTpes}
	Let $B\in \F_q^{n\times b}$ be a matrix whose columns form a $\RC \rho \ell$ set. Denote $\tau = \distr B$. Then, $H_q(\tau) \le b\cdot h_q(\ell,\rho)+\ell$, where $h_{q,\ell}(\rho) = \rho\log_q \inparen{\frac {q-\ell}\rho} + (1-\rho)\log_q \inparen{\frac \ell{1-\rho}}$.
\end{claim}
\begin{proof}
	Let $Z_1,\dotsc, Z_n$ be subsets of $\F_q$, each of size $\ell$, such that for all $j \in [b]$, we have $\inabset{i\in [n]\mid B_{i,j}\notin Z_i} \le \rho n$. Let $i$ be sampled uniformly from $[n]$.  We now have
	$$H_q(\tau) = H_q(B_i) \le H_q(B_i, Z_i) = H_q(Z_i)+H_q(B_i\mid Z_i) \le H_q(Z_i) + \sum_{j=1}^b H_q(B_{i,j}\mid Z_i).$$
	The number of different options for $Z_i$ is $\binom q\ell$ so $H_q(Z_i)\le \log_q(\binom q\ell)\le \ell$ . Let $\rho'_j$ ($j\in [b]$) denote the probability that $B_{i,j}\notin Z_i$, and note that $\rho'_j \le \rho$. Then,
	$$H_q(B_{i,j}\mid Z_i) \le h_{q,\ell}(\rho'_j) \le h_{q,\ell}(\rho).$$
	Consequently, $H_q(\tau) \le b\cdot h_{q,\ell}(\rho)+\ell$, establishing the claim.		
\end{proof}

\begin{proof}[Proof of Proposition \ref{prop:RLCListRecoveryCapacity}]
	\sloppy Let $\cP$ denote the property consisting of codes over $\F_q^n$ that are \textbf{not} $\LR \rho \ell L$.  
	Let $\tau \in \cT_{\cP}$ and let $A\in \cM_{n,\tau}$ be a matrix in $\F_q^{n\times a}$ ($a\in \N$). By Observation \ref{obs:ListRecoveryGoodProperty}, $A$ is $\RSC \rho \ell {L+1}$. Let $W$ be a $(\rho,\ell)$-recovery-clustered set of size $L+1$, contained in the column-span of $A$.  Note that $W$ must contain a linearly-independent subset $U$ of size $b:=\ceil{\log_q |W|}=\ceil{\log_q (L+1)}$. Let $D\in \F_q^{a\times b}$ such that $B:=AD$ is the matrix whose columns are the elements of $U$, and note that $U$ is also $\RC \rho \ell$. Let $\tau' = \distr{B}$. By Claim \ref{claim:HBoundForRCTpes}, $H_q(\tau')\le b\cdot h_{q,\ell}(\rho) + \ell$. Furthermore, we can express $\tau'$ as the distribution of the random vector $xD$, where $x\sim \tau$. Consequently, $\tau'\in \cI_{\tau}$. Therefore,
	$$\max_{\tau''\in \cI_\tau}\inparen{1-\frac{H_q(\tau'')}{\dim(\tau')}} \ge 1-\frac{H_q(\tau')}{b} \ge 1-h_{q,\ell}(\rho) - \frac{\ell}{b}.$$
	The claim now follows by Theorem \ref{thm:RLCThreshold}.
\end{proof}

We note that the above derivation of Proposition \ref{prop:RLCListRecoveryCapacity} could also be achieved via more standard arguments, which do not require Theorem \ref{thm:RLCThreshold}. The actual power of Theorem \ref{thm:RLCThreshold} is that it enables reductions from other random code models to the RLC model, as demonstrated in the proof of Theorems \ref{thm:MainLowBiasProperties} and \ref{thm:MainLargeDistanceProperties}, via Lemma \ref{lem:ReductionToRLC}. This sort of argument involves an application of Theorem \ref{thm:RLCThreshold} in \emph{its less intuitive direction}: rather than starting from an upper bound on $H_q(\tau)$ for some set of distributions and using Theorem \ref{thm:RLCThreshold} to obtain a lower bound on $\TRLC(\cP)$, we start from some known lower bound on $\TRLC(\cP)$ and use the theorem to get an upper bound on the entropy of certain ``bad distributions.'' The latter entropy bound is then typically used in a union-bound argument to obtain a lower bound on the threshold rate for some non-RLC model. This type of argument was used in \cite{MRRSW} to prove that \emph{Gallager LDPC codes} are as list-decodable (and list-recoverable) as RLCs.

\begin{remark}[Average versions of list-decodability and list-recoverability]
	Average-radius list-decodability is a stronger property where we demand that for every $L+1$ codewords their average distance to any center exceeds $\rho$ (as opposed to maximum distance for list-decodability). A code not being $(\rho,L)$-average-radius list-decodable is also an $(L+1)$-local, row-symmetric and scalar-invariant property.
	
	For list-recovery, we can define a stronger variant where in Definition \ref{def:LR} we allow input sets $Z_i$ such that the average size $|Z_i|$ over all $i \in [n]$ is at most $\ell$. A violation of this stronger property is also a local, row-symmetric and scalar-invariant property.
	
	The generality of our framework thus means that we can get results for these variants also automatically.  We note that certain results for list-decodability for RLCs, e.g., \cite{GHK,LiWootters}, do not extend to average-radius list-decoding (or list-recovery).
\end{remark}

\subsection{Reducing \emph{local similarity to RLC} to \emph{local similarity to RLC in expectation}}

The following lemma uses Theorem \ref{thm:RLCThreshold} to show that any random code $\cC$ is \emph{locally-similar to an RLC}, provided that $\cC$ is similar to an RLC \emph{in expectation}---a term made precise in Eq.\ \eqref{eq:GenericReductionHyptohesis}.

\begin{lemma}[A generic reduction to random linear codes]\label{lem:ReductionToRLC}
	Let $n\in \N$, $q$ a prime power and $b\in \N$ such that $\frac{n}{\log_q n}\ge \omega_{n\to \infty}\inparen{q^{2b}}$. Let $\cC\in \F_q^{n}$ be a linear code of rate $R\in [0,1]$, sampled at random from some ensemble.
	Suppose that, for every $1\le a\le b$ and every distribution $\tau$ over $\F_q^a$, we have
	\begin{equation}\label{eq:GenericReductionHyptohesis}
		\Eover{\cC}{\inabset{A\in \cM_{n,\tau} \mid A\subseteq \cC}} \le q^{(H_q(\tau)-a(1-R)+a\eps)n}\enspace,
	\end{equation}
	for some fixed $\eps > 0$.
	Then, for any row-symmetric and $b$-local property $\cP$ over $\F_q^n$ such that $R\le \TRLC(\cP)-2\eps$, it holds that
	$$\PROver{\cC}{\cC\text{ satisfies }\cP} \le q^{-n\inparen{\eps-o_{n\to\infty}(1)}}\enspace . $$
\end{lemma}
\begin{remark}
	The expected number of $\tau$-distributed matrices in an RLC of rate $R$ is $$\inabs{\cM_{n,\tau}} \cdot q^{n(R-1)b}\approx q^{n\inparen{H_q(\tau)-(1-R)b}}\enspace . $$
	Thus, Eq.\ \eqref{eq:GenericReductionHyptohesis} essentially says that the expected number of $\tau$-distributed matrices in $\cC$ is not much larger than in an RLC of similar rate.
\end{remark}

\begin{proof}
	Let $\tau\in \cT_\cP$. By Theorem \ref{thm:RLCThreshold}, there is some distribution $\tau'\in \cI_\tau$ over $\F_q^a$ (where $1\le a\le b$) such that $$\frac{H_q(\tau')}{a}\le 1-\TRLC(\cP) +o_{n\to\infty}(1)\enspace . $$ Now, by Observation \ref{obs:ImpliedMotivation}, followed by Markov's bound,
	\begin{align*}
		\PR{\exists A\in \cM_{n,\tau} ~~A\subseteq \cC} &\le \PR{\exists A\in \cM_{n,\tau'} ~~A\subseteq \cC}  \le \E{\inabset{A\in \cM_{n,\tau'}\mid A\subseteq \cC}} \\ &\le \expOver q{\inparen{H_q(\tau')-a(1-R)+a\eps}n} \\ &\le \expOver q{an\inparen{R-\TRLC(\cP)+\eps +o_{n\to\infty}(1)} } \\ 
		&\le \expOver q{-na(\eps-o_{n\to\infty}(1))}\enspace.
	\end{align*} 
	Therefore, by Fact \ref{fact:PropertyDecomposition},
	\begin{align*}
		\PR{\cC\text{ satisfies }\cP} &\le \sum_{\tau\in \cT_\cP} \PR{\exists A\in \cM_{n,\tau}~~A\subseteq \cC} \le \inabs{\cT_{\cP}}q^{-n(\eps-o_{n\to\infty}(1))} \\&\le (n+1)^{q^b}q^{-n(\eps-o_{n\to\infty}(1))} \le q^{-n(\eps-o_{n\to\infty}(1))} \enspace. \qedhere
	\end{align*}
\end{proof}

\section{Random puncturings of certain codes are locally-similar to RLCs}\label{sec:SimilarityToRLC}

In this section we prove Theorems \ref{thm:MainLowBiasProperties} and \ref{thm:MainLargeDistanceProperties}, restated below.

\MainLowBiasProperties*
\MainLargeDistanceProperties*

\subsection{Proofs outlines for Theorems \ref{thm:MainLargeDistanceProperties} and \ref{thm:MainLowBiasProperties}}	
To prove Theorems \ref{thm:MainLowBiasProperties} and \ref{thm:MainLargeDistanceProperties} using Lemma \ref{lem:ReductionToRLC}, we need to show that the code ensembles considered in these lemmas are indeed similar in expectation to an RLC. The respective  claims for each theorem are stated below as Lemmas \ref{lem:tauInLowBias} and \ref{lem:tauInLargeDistance}.

\begin{lemma}[Puncturings of low-bias codes are locally similar in expectation to random linear codes]\label{lem:tauInLowBias}
	Fix $b\in \N$ and a full-rank distribution $\tau$ over $\F_q^b$. Let $\cD\subseteq \F_q^m$ be an $\eta$-biased linear code ($\eta \ge 0$). Let $\varphi$ be a random ($m\to n$) puncturing map and let $\cC = \varphi(\cD)$. Denote $R = \frac{\log_q|\cD|}n$. Then,
	$$\Eover{\cC}{\inabset{A\in \cM_{n,\tau}\mid A\subseteq \cC}} \le \expOver q{n\cdot \inparen{H_q(\tau)-(1-R)b + \log_q\inparen{1+\eta q^b}}}$$
\end{lemma}

\begin{lemma}[Puncturings of large-distance codes are locally similar in expectation to random linear codes]\label{lem:tauInLargeDistance}
	Fix $b\in \N$ and a full-rank distribution $\tau$ over $\F_q^b$. Let $\cD\subseteq \F_q^m$ be a linear code of $\eta$-optimal distance ($\eta \ge 0$). Let $\Lambda \sim \uniform(\diagonals_n)$ and, independently, let $\varphi$ be a random $(m\to n)$ puncturing map. Denote $R = \frac {\log_q |\cD|}n$. Then,
	$$\E{\inabset{A\in \cM_{n,\tau}\mid A\subseteq \Lambda\cdot  \varphi(\cD)}} \le \expOver q{n\cdot \inparen{H_q(\tau)-(1-R)b + \log_q\inparen{1+\eta q^b}+\log_q 2}}$$
\end{lemma}

We defer the proofs of these lemmas to Section \ref{sec:LargeDistanceApproximatesRLC}, and first show that they imply Theorems \ref{thm:MainLowBiasProperties} and \ref{thm:MainLargeDistanceProperties}.

\begin{proof}[Proof of Theorem \ref{thm:MainLowBiasProperties}]
	Let $\tau$ be a distribution over $\F_q^a$ with $a\le b$. Lemma \ref{lem:tauInLowBias} yields
	\begin{align*}
		\Eover{\cC}{\inabset{A\in \cM_{n,\tau}\mid A\subseteq \cC}} &\le \expOver q{n\cdot \inparen{H_q(\tau)-(1-R)a + \log_q\inparen{1+\eta q^a}}} \\&\le \expOver q{n\cdot \inparen{H_q(\tau)-(1-R)a + \frac{\eta q^a}{\ln q}}}\\&= \expOver q{{n\cdot \inparen{H_q(\tau)-(1-R)a + \frac{\eps b}{q^{b-a}}}}}\\ &\le \expOver q{{n\cdot \inparen{H_q(\tau)-(1-R)a + a\eps}}}\enspace,
	\end{align*}
	which implies the claim by virtue of Lemma \ref{lem:ReductionToRLC}.
\end{proof}

\begin{proof}[Proof of Theorem \ref{thm:MainLargeDistanceProperties}]
	Let $\tau$ be a distribution over $\F_q^a$, where $a\le b$. By Lemma \ref{lem:tauInLargeDistance}, 
	\begin{align*}
		\Eover{\cC}{\inabset{A\in \cM_{n,\tau}\mid A\subseteq \Lambda \cC}} &\le \expOver q{n\cdot \inparen{H_q(\tau)-(1-R)a + \log_q\inparen{1+\eta q^a}+\log_q 2}} \\&\le \expOver q{n\cdot \inparen{H_q(\tau)-(1-R)a + 2\log_q2}} \\&\le \expOver q{n\cdot \inparen{H_q(\tau)-(1-R)a + \eps}} \\
		&\le \expOver q{n\cdot \inparen{H_q(\tau)-(1-R)a + a\eps}} \enspace,
	\end{align*}
	where $\Lambda \sim \uniform(\diagonals_n)$. By Lemma \ref{lem:ReductionToRLC}, $\Lambda\cC$ satisfies $\cP$ with probability at most $q^{-n(\eps-o_{n\to\infty}(1))}$. Since $\cP$ is scalar-invariant, the same holds for $\cC$.
\end{proof}

\subsection{Certain code ensembles are similar in expectation to an RLC --- Proofs of Lemmas \ref{lem:tauInLowBias} and \ref{lem:tauInLargeDistance}}
\label{sec:LargeDistanceApproximatesRLC}
Our goal in this section is to prove Lemmas \ref{lem:tauInLowBias} and \ref{lem:tauInLargeDistance}, which follow from Lemmas \ref{lem:Vazirani}, \ref{lem:SmoothDistribution} and \ref{lem:ProbOfHavingGivenType}, stated and proven below. The proof of Lemmas \ref{lem:tauInLargeDistance} and \ref{lem:tauInLowBias} is then completed at the end of this section.

\subsubsection{Empirical distributions of sub-sampled matrices	}

Lemma \ref{lem:Vazirani} is a variation of the \emph{Vazirani XOR-Lemma} (see \cite{3XOR}, and Lemma \ref{lem:VaziraniOriginal} for a special case). Given a distribution $\sigma$ over $\F_q^b$ , the XOR-Lemma relates the \emph{total-variation distance} of $\sigma$ from the uniform distribution over $\F_q^b$, to the maximum of $\inabs{\wh \sigma(y)}$ over all $y\ne 0$. In Lemma \ref{lem:Vazirani}, rather than taking a maximum, we consider the $\ell_1$ norm of $\wh \sigma$, which yields a tighter bound when only a small number of entries of $\wh \sigma$ are large in absolute value.

\begin{lemma}\label{lem:Vazirani}
	Fix a prime power $q$, and $b\in \N$. Let $\sigma$ be a distribution over $\F_q^L$ and let $f:\F_q^b\to \R$ be a non-negative function. Then,
	$$\Eover{x\sim \sigma}{f(x)} \le \inparen{\sum_{y\in \F_q^b}\inabs{\wh \sigma(y)}}\cdot \Eover{x\sim \uniform(\F_q^b)}{f(x)}\enspace . $$
\end{lemma}
\begin{proof}
	We have 
	$$\sigma(x) = q^{-b}\sum_{y\in \F_q^b} \wh \sigma(y) \omega^{-\tr(\ip xy)} \le q^{-b}\sum_{y\in \F_q^b}\inabs{\wh \sigma(y)}$$
	for all $x\in \F_q^b$. So 
	$$ \Eover{x\sim \sigma}{f(x)}  = \sum_{x \in \F_q^b} \sigma(x) f(x) \le  q^{-b} \inparen{ \sum_{y\in \F_q^b}\inabs{\wh \sigma(y)} }\cdot \inparen{ \sum_{x\in \F_q^b}f(x) }=\inparen{\sum_{y\in \F_q^b}\inabs{\wh \sigma(y)}}\cdot \Eover{x\sim \uniform(\F_q^b)}{f(x)}\enspace. \qedhere $$
\end{proof}

We next bound the expectation of an arbitrary non-negative test function over the empirical row-distribution of a given matrix $B$, assuming that the column-span of $B$ has good bias or distance. The bias based bound is an immediate application of Lemma \ref{lem:Vazirani}. The weight based bound requires an additional trick, and only yields a result relating to the row-distribution of $B^*$ rather than $B$ itself (recall Definition \ref{def:ScalarExpansion} for a reminder about $B^*$). One reason for the difference between the two cases is that under the weight-based hypothesis we have an upper bound only on the entries of the Fourier transform (Eq.\ \eqref{eq:Sigma_B^*OneSidedBound}), rather than on their absolute value. 

\begin{lemma}\label{lem:SmoothDistribution}
	Let $B\in \F_q^{m\times b}$ have $\rank B = b$, and let $f:\F_q^b\to \R$ be a non-negative function. Then, the following holds for all $\eta \ge 0$:
	\begin{enumerate}
		\item \label{enum:SmoothDistributionLowBias}
		Suppose that the column-span of $B$ (as a code in $\F_q^m$) is $\eta$-biased. Then,
		$$\Eover{x\sim \distr B}{f(x)} \le(1+q^b\eta)\cdot \Eover{x\sim \uniform(\F_q^b)}{f(x)}\enspace . $$
		\item \label{enum:SmoothDistributionGoodDistance}
		Suppose that the column-span of $B$ has $\eta$-optimal distance. Then, $$\Eover{x\sim \distr {B^*}}{f(x)}\le 2(1+q^b\eta)\cdot \Eover{x\sim \uniform(\F_q^b)}{f(x)}\enspace . $$
	\end{enumerate}
	
\end{lemma}
\begin{proof}
	We first prove Item \ref{enum:SmoothDistributionLowBias}. Hence, by Lemma \ref{lem:Vazirani} it suffices to show that 
	$$
	\sum_{y\in \F_q^b}\inabs{\wh{\distr B}(y)} \le 1+q^b\eta\enspace.
	$$
	By Eq.\ \eqref{eq:FourierMatrixImage}, the above is equivalent to 
	\begin{equation}\label{eq:SmoothDistributionLowBiasNeedToShow}
		\sum_{y\in \F_q^b}\inabs{\wh{\distr {By}}(1)} \le 1+q^b\eta\enspace.
	\end{equation}
	For $y=0$ we have $\wh{\distr{By}}(1) = \wh{\distr 0}(1) = 1$. For any $y\in \F_q^b\setminus \{0\}$, since $B$ has full column-rank, $By$ is a non-zero codeword of $\cD$. By hypothesis, $By$ is $\eta$-biased, so Fact \ref{fact:BiasAndFourier} yields $\inabs{\wh{\distr {By}}(1)} \le \eta$, establishing Eq.\ \eqref{eq:SmoothDistributionLowBiasNeedToShow}.
	
	We now turn to Item \ref{enum:SmoothDistributionGoodDistance}. Let $\sigma$ denote the distribution, over $\F_q^b$, of the random variable $a\cdot x$, where $a\sim \uniform(\F_q^*)$ and $x$ is independently sampled from $\distr B$. By Lemma \ref{lem:Vazirani}, to prove Item \ref{enum:SmoothDistributionGoodDistance} it suffices to show that
	
	\begin{equation}\label{eq:SmoothDistributionGoodDistanceNeedToShow}
		\sum_{y\in \F_q^b}	\inabs{\wh{\sigma}(y)} \le 2 \inparen{1+q^b\eta}\enspace.
	\end{equation}
	
	By Eq.\ \eqref{eq:WeightFourier} followed by Eq.\ \eqref{eq:FourierMatrixImage},
	$$\weight(By) = \frac{q-1}q - \frac1q\cdot \sum_{a\in \F_q^*}\wh{\distr {By}}(a) = \frac{q-1}q - \frac1q\cdot \sum_{a\in \F_q^*}\wh{\distr B}(ay) = \frac{q-1}q\cdot \inparen{1-\wh \sigma(y)}\enspace,$$
	so $\wh\sigma(y) = 1-\frac{q}{q-1}\cdot \weight(By)$.
	
	In particular, if $y\ne 0$ then $By$ is a non-zero element in the column-span of $B$. Hence, by hypothesis,
	\begin{equation}\label{eq:Sigma_B^*OneSidedBound}
		\wh{\sigma}(y) = 1-\frac{q}{q-1}\cdot \weight(By) \le \eta\enspace.
	\end{equation}

	Let $P = \{y\in \F_q^b\mid \wh{\sigma}(y) \ge 0\}$ and $N = \F_q^b\setminus P$. By Eq.\ \eqref{eq:Sigma_B^*OneSidedBound}, 
	$$\sum_{y\in P\setminus\{0\}} \wh{\sigma}(y) \le q^b\eta\enspace . $$
	Note that $\wh{\sigma}(0) = \sum_{x\in \F_q^b}\sigma(x) = 1$, and thus,
	$$\sum_{y\in P} \wh{\sigma}(y) = 1 +  \sum_{y\in P\setminus\{0\}} \wh{\sigma}(y) \le 1+q^b\eta\enspace . $$
	Consequently,
	$$0 \le q^b\cdot \sigma(0) = \sum_{y\in \F_q^b} \wh{\sigma}(y) = \sum_{y\in P}\wh{\sigma}(y) + \sum_{y\in N}\wh{\sigma}(y) \le 1+q^b\eta+\sum_{y\in \N}\wh{\sigma}(y)$$
	and so,
	$$\sum_{y\in N}\inabs{\wh{\sigma}(y)} = -\sum_{y\in N}\wh{\sigma}(y) \le 1+ q^b\eta\enspace . $$
	Eq.\ \eqref{eq:SmoothDistributionGoodDistanceNeedToShow} now follows since
	$$\sum_{y\in \F_q^b}\inabs{\wh{\sigma}(y)} = \sum_{y\in P}\inabs{\wh{\sigma}(y)} + \sum_{y\in N}\inabs{\wh{\sigma}(y)} \le 2(1+q^b\eta) \enspace. \qedhere $$
\end{proof}

Lemma \ref{lem:ProbOfHavingGivenType} bounds the probability of a random puncturing of a given matrix $B$ having a certain empirical distribution $\tau$. Due to the concavity argument in Eq.\ \eqref{eq:KL-bound}, this lemma gives tighter bounds when $\distr B$ is close to the uniform distribution over $\F_q^b$. Notably, as Lemma \ref{lem:SmoothDistribution} shows, good bias or similar properties of the column-span of $B$ ensure that $\distr B$ is indeed close to uniform.

\begin{lemma}\label{lem:ProbOfHavingGivenType}
	Fix some distribution $\tau$ over $\F_q^b$. Let $B\in \F_q^{m\times b}$ have $\rank B = b$. Let $\varphi:\F_q^m\to \F_q^n$ be a random puncturing map.  Then,
	$$\PR{\varphi(B)\in \cM_{n,\tau}} \le \expOver{q}{n\inparen{\log_q\Eover{x\sim \distr B}{\tau(x)}+H_q(\tau)}}\enspace . $$
\end{lemma}
\begin{proof}
	By Fact \ref{fact:DKL},
	\begin{equation}
		\label{eq:tailbound}
		\PR{\varphi(B)\in \cM_{n,\tau}} = \PR{\distr{\varphi(B)} = \tau}\le  q^{-n\cdot \DKL \tau {\distr {B}}  q}\enspace.
	\end{equation}
	By concavity of $\log$,
	\begin{align*}	
		\DKL \tau {\distr {B}} q &= \sum_{x\in \F_q^b} \tau(x)\log_q\frac{\tau (x)}{{\distr {B}}(x)} = -H_q(\tau) - \sum_{x\in \F_q^b}\tau(x)\log_q {\distr {B}}(x)\\ &\ge -H_q(\tau)-\log_q \Eover{x\sim \distr B}{\tau(x)}\enspace.\numberthis \label{eq:KL-bound}
	\end{align*}
	The claim follows from Eq.\ \eqref{eq:tailbound} and Eq.\ \eqref{eq:KL-bound}.
\end{proof}

\subsubsection{Proofs of Lemmas \ref{lem:tauInLowBias} and \ref{lem:tauInLargeDistance}}

\begin{proof}[Proof of Lemma \ref{lem:tauInLowBias}]
	Let $\tau$ be a full-rank distribution over $\F_q^b$. Item \ref{enum:SmoothDistributionLowBias} of Lemma \ref{lem:SmoothDistribution} yields
	$$\Eover{x\sim \distr B}{\tau(x)}\le q^{-b}\inparen{1+q^b\eta}\enspace , $$
	for all $B\in \F_q^{m\times b}$ such that $\rank B = b$ and $B\subseteq \cD$. By Lemma \ref{lem:ProbOfHavingGivenType},
	$$\PR{\varphi(B)\in \cM_{n,\tau}}\le \expOver{q}{n\inparen{-b+H_q(\tau)+\log_q\inparen{1+q^b\eta}}}\enspace.$$
	The claim now follows by the union bound over the $\le q^{Rnb}$ choices of $B$.
\end{proof}

\begin{proof}[Proof of Lemma \ref{lem:tauInLargeDistance}]
	By Observation \ref{obs:LambdaCAsAPuncturing}, 
	$\Lambda\cdot \varphi(\cD)$ is distributed identically to $\varphi^*(\cD^*)$, where $\varphi^*$ is a random $((q-1)m\to n)$ puncturing map. Thus,
	\begin{align}\E{\inabset{A\in \cM_{n,\tau}\mid A\subseteq \Lambda\cdot  \varphi(\cD)}} &= \E{\inabset{A\in \cM_{n,\tau}\mid A\subseteq \varphi^*(\cD^*)}} \nonumber\\ &\le \E{\inabset{B\in \F_q^{m\times b}\mid B\subseteq \cD \text{ and } \varphi^*(B^*) \in \cM_{n,\tau}}}\enspace. \label{eq:TauInLargeDistanceNTP}
	\end{align}
	We proceed to bound the expectation of the right-hand side. 
	
	Suppose that $\varphi^*(B^*)\in \cM_{n,\tau}$. Because $\tau$ is of full-rank, we have $\rank B = \rank B^* \ge \rank \varphi^*(B^*) = b$, so $\rank B = b$. 
	
	Let $B\in \F_q^{m\times b}$ such that $\rank B = b$ and $B\subseteq \cD$. Since the column-span of $B$ is contained in $\cD$, it is of $\eta$-optimal distance. Hence, by Item \ref{enum:SmoothDistributionGoodDistance} of Lemma \ref{lem:SmoothDistribution}, 
	\begin{equation}\label{eq:TauInLargeDistanceMBound}
		\Eover{x\sim \distr {B^*}}{\tau(x)} \le \Upsilon \enspace,
	\end{equation}
	where $\Upsilon = 2q^{-b}\inparen{1+q^b\eta}$. Lemma \ref{lem:ProbOfHavingGivenType} yields
	\begin{align*}
		\E{\inabset{B\in \F_q^{m\times b}\mid B\subseteq \cD \text{ and } \varphi^*(B^*) \in \cM_{n,\tau}}} &= \sum_{\substack {B\in \F_q^{m\times b}\\ B\subseteq \cD\\ \rank B = b}} {\PROver{\Lambda,\varphi}{\varphi^*(B^*)\in \cM_{n,\tau}}}\\ 
		&\le q^{bRn}\cdot  
		\expOver q{n\inparen{\log_q \Upsilon +H_q(\tau)}}\enspace,\numberthis \label{eq:TauInLargeDistanceExpBound}
	\end{align*}
	and the claim follows from Eqs.\ \eqref{eq:TauInLargeDistanceNTP}, \eqref{eq:TauInLargeDistanceMBound} and \eqref{eq:TauInLargeDistanceExpBound}.
\end{proof}

\section{A random puncturing of a near-optimal-distance code is likely to be list-decodable}\label{sec:ProofOfMainReductToGHK}	
Our goal in this section is to prove Theorem \ref{thm:MainReduceToGHK} on the list-decodability of random puncturings of any mother code of sufficiently high distance.

\subsection{GHK list-decodability bound for random linear codes revisited}
The main result of \cite{GHK}  gives bounds on the list-size for list-decoding of RLCs up to capacity. Here, we go deeper and slightly reformulate\footnote{See Remark \ref{rem:GHK} for the differences in our formulation.}
the main technical claim of that paper.

\begin{theorem}[{\cite[Thm.\ 6.1]{GHK}}]\label{thm:GHKTechnical}
	Let $q$ be a prime power and let $\rho \in (0,1-1/q)$. Then, there is a constant $K' = K'_{\rho,q} \ge 1$ such that, for all $b,L\in \N$, we have
	\begin{equation*}
		\inabset{A\in \F_q^{n\times b }\mid A\text{ is $(\rho,L+1)$-span-clustered}} \le  q^{(b h_q(\rho)-4)\cdot n}
	\end{equation*}
	whenever $L \ge K'\cdot b$ and $n$ is large enough, and 
	\begin{equation}\label{eq:cXLargeEllBound}
		\inabset{A\in \F_q^{n\times b }\mid A\text{ is $(\rho,L+1)$-span-clustered}} \le  q^{(b h_q(\rho)+1)\cdot n}
	\end{equation}
	in general.
	
	Furthermore, one can take 
	\begin{equation}\label{eq:GHKCAsymptotics}
		K' \le \exp{O\inparen{\frac {(\log_2 q)^2}{\min\inset{(1-1/q-\rho)^2,\rho}}}}\enspace.
	\end{equation}
\end{theorem}

\begin{remark}
	\label{rem:GHK}
	There are several differences between our formulation of the theorem and the one that appears in \cite{GHK}. We list and justify them here. 
	\begin{enumerate} 
		\itemsep=0ex
		\item[(i)] The random vectors $X_1,\dotsc, X_\ell$ from the original formulation have become the columns of the matrix $A$, and we changed the name $\ell$ to $b$.
		\item[(ii)] The original statement of \cite[Thm.\ 6.1 ]{GHK} only deals with matrices whose span contains a large set clustered around $0$. In our statement we already apply the reduction to a ball with arbitrary center, which appears in \cite[Thm.\ 2.1]{GHK}.
		\item[(iii)] Eq.\ \eqref{eq:cXLargeEllBound} is a rather naive bound, originally derived as part of the proof of \cite[Thm.\ 2.1]{GHK}.
		\item[(iv)] The asymptotic statement about $K'_{\rho,q}$ comes from inspecting the proof in \cite{GHK}. Specifically, in the notation of that paper, \cite[Lem.\ 6.3]{GHK} yields a $2$-increasing chain of length $d = \Omega(\log_q K')$ whenever $L \ge K'\cdot b$. The exponent in the $q$-ary analog of \cite[Lem 4.1]{GHK} satisfies $\delta_p = \Theta\inparen{\tfrac{\min\inset{\rho,\inparen{1-\frac 1q-\rho}^2}}{\log_2 q}}$. Finally, the requirement in \cite[Thm.\ 6.1]{GHK} is that $K'$ be large enough so that $d\cdot \delta_p \ge \Omega(1)$.
	\end{enumerate}
\end{remark}

It will be convenient to formulate a corollary from Theorem \ref{thm:GHKTechnical} in terms of span-clustered distributions (recall Definition \ref{def:SpanClusteredType}).

%	\begin{definition}\label{def:SpanClusteredType}
	%		Fix a prime power $q$. Let $b,n\in N$ and let $\tau$ be an $n$-feasible distribution over $\F_q^b$. If a matrix $A\in \cM_{n,\tau}$ is $(\rho,L+1)$-span-clustered, we say that $\tau$ is \deffont{$(\rho,L+1)$-span-clustered (with regard to $n$)}.
	%	\end{definition}
%	\begin{remark}\label{rem:SpanClusteredRowSymmetric}
	%		Observe that the notion of $\tau$ being $(\rho,L+1)$-span clustered is well defined, and in particular does not depend on the choice of $A$ in Definition \ref{def:SpanClusteredType}. In other words, either every matrix in $\cM_{n,\tau}$ is $(\rho,L+1)$-span-clustered, or none of them are. Indeed, suppose that $A\in \F_q^{n\times b}$ is $(\rho,L+1)$-span-clustered with regard to some center $z\in \F_q^n$, and let $B$ be a matrix obtained from $A$ by permuting the rows of the latter according to some permutation $\pi$ over $[n]$. Then, $B$ is $(\rho,L+1)$-span-clustered with regard to the center vector resulting from applying $\pi$ to $z$. 
	%		
	%		The above observation is a result of the fact that containing a $\SC \rho{L+1}$ matrix is a \deffont{row-symmetric} code property (see Definition \ref{def:LocalRowSymmetric}).
	%	\end{remark}

\begin{corollary}\label{cor:GHKEntropy}
	In the setting of Theorem \ref{thm:GHKTechnical}, every $(\rho,L+1$)-span-clustered (with regard to $n$), $n$-feasible distribution $\tau$ over $\F_q^b$ satisfies
	\begin{equation*}
		H_q(\tau) \le b\cdot \inparen{h_q(\rho)+\frac{5K'_{\rho,q}}L} - 3
	\end{equation*}
	for every $b$ and $n$ such that $\frac{n}{\log_qn}\ge \omega\inparen{q^{L+1}}$. 
\end{corollary}
\begin{proof}
	Since $\tau$ is $\SC{\rho}{L+1}$, $\cM_{n,\tau} \subseteq \inset{A\in \F_q^{n\times b }\mid A\text{ is $(\rho,L+1)$-span-clustered}}$. Thus, Eq.\ \eqref{eq:M_nTauApproximation} and Theorem \ref{thm:GHKTechnical} yield the following:
	
	\paragraph{If $L \ge K'_{\rho,q}\cdot b$:}
	$$
	H_q(\tau) \le \frac{\log_{q}\inabs{\cM_{n,\tau}}}{n} + O\inparen{\frac{q^b\cdot \log_q n}n} \le  bh_q(\rho)-4 + O\inparen{\frac{q^b\cdot \log_q n}n}
	$$
	\paragraph{If $L < K'_{\rho,q}\cdot b$:}
	\begin{align*}
		H_q(\tau) &\le \frac{\log_{q}\inabs{\cM_{n,\tau}}}{n} + O\inparen{\frac{q^b\cdot \log_q n}n} \le  bh_q(\rho)+1 +  O\inparen{\frac{q^b\cdot \log_q n}n} \\ &\le  bh_q(\rho) + \frac{5bK'_{\rho,q}}L - 4 + O\inparen{\frac{q^b\cdot \log_q n}n}\enspace.		
	\end{align*}
	The claim now follows from our assumption that $\frac{n}{\log_q n} \ge \omega\inparen{q^{L+1}}$.
\end{proof}

\subsection{Proof of Theorem \ref{thm:MainReduceToGHK}}
We now turn to proving Theorem \ref{thm:MainReduceToGHK}, restated here.
\mainReduceToGHK*

Before proving the theorem, we compare it to several known results about list-decodability of RLCs. By the List-Decoding Capacity Theorem, Theorem \ref{thm:MainReduceToGHK} achieves the optimal trade-off between $q$, $\rho$ and $R$. We thus turn to discuss the secondary trade-off, which involves the former three parameters and the \emph{list-size} $L$.  As mentioned in Section \ref{sec:ResultsGHK}, Theorem \ref{thm:MainReduceToGHK} is derived by reduction to the result of \cite{GHK} on list-decodability of RLCs. The main theorem of \cite{GHK} states that a RLC of rate $R = 1-h_q(\rho)-\frac{K'_{\rho,q}}{L}$ is with high probability is $\LD \rho L$, 
where $K'_{\rho,q} \le \exp{O\inparen{\frac {(\log q)^2}{\min\inset{(1-1/q-\rho)^2,\rho}}}}$ is proportional to the constant $K_{\rho,q}$ that appears in Theorem \ref{thm:MainReduceToGHK}. Denoting the \emph{gap-to-capacity of the rate} by $\eps = 1-h_q(\rho)-R$, \cite{GHK} shows that an RLC of rate $R$ is almost surely $\LD \rho L$ with $L \approx \frac{K'_{\rho,q}}\eps$. In Theorem \ref{thm:MainReduceToGHK}, we have $\eps = \frac{K_{\rho,q}}L$, so $L = \frac{K_{\rho,q}}\eps = O\inparen{\frac {K'_{\rho,q}}\eps}$. Thus, we can informally state Theorem \ref{thm:MainReduceToGHK} as ``A random puncturing of a code of near-optimal distance is very likely to be list-decodable up to capacity, with a similar list-size trade-off to that guaranteed by \cite{GHK} for RLCs.''

The list-size $L$ guaranteed by Theorem \ref{thm:MainReduceToGHK} inherits some desirable properties from \cite{GHK}: it is constant in terms of $n$, and has linear dependence on $\frac 1\eps$, which is tight for RLCs \cite[Thm.\ 16]{GN14}. As for the dependence on $q$ and $\rho$, we get good list-size bounds when $q$ is not too large and $\rho$ is bounded away from $0$ and $1-\frac 1q$, but, unfortunately, the constant $K_{\rho,q}$ grows exponentially as $\rho\to 1-\frac 1q$. In comparison with \cite{GHK}, other works on RLC list-decodability are more specialized, and give tighter upper bounds on the list-size in specific regimes. Notably, \cite{Wootters13} does well when  $\rho$ is large and $\eps$ is of similar magnitude to $R$, and \cite{LiWootters} gives an extremely tight upper bound (see \cite{GLMRSW}) on the list-size for every $\rho$ and $\eps$, when $q=2$.

We note that, while Theorem \ref{thm:MainReduceToGHK} only achieves the analogue of \cite{GHK} for randomly punctured codes, Theorems \ref{thm:MainLowBiasProperties} and \ref{thm:MainLargeDistanceProperties}, with their somewhat stronger hypotheses, achieve (in particular) analogues of all known \emph{and future} positive results about list-decodability of RLCs. The obstacle to concluding such a broad result solely from the hypothesis of Theorem \ref{thm:MainReduceToGHK} is discussed in Remark \ref{rem:ReasonForConditions}.

\begin{proof}[Proof of Theorem \ref{thm:MainReduceToGHK}]
	Take $K_{\rho,q} = 5K'$, where $K'_{\rho,q}$ is as in Theorem \ref{thm:GHKTechnical}. We need to show that
	$$
	\PR{\cC\text{ is not }\LD \rho L} \le q^{-\Omega(\eps n)}.
	$$
	Since being not $\LD \rho L$ is a scalar-invariant property (see Definition \ref{def:ScalarInvariant} and Section \ref{sec:ScalarExpandedCode}), it suffices to show instead that
	\begin{equation}\label{eq:MainReduceToGHKNeedToProve}
		\PR{\Lambda\cC\text{ is not }\LD \rho L} \le q^{-\Omega(\eps n)}\enspace,
	\end{equation}
	where the matrix $\Lambda$ is sampled uniformly from $\diagonals_n$.		
	
	Now, if $\Lambda\cC$ is \emph{not} $\LD \rho L$, then $\Lambda\cC$ contains some $(\rho,L+1)$-span-clustered matrix $A\in \F_q^{n\times b}$ for some $b$, $\log_q(L+1) \le b\le L+1$. Hence,
	\begin{align*}
		&\PR{\Lambda\cC\text{ is not }\LD \rho L} \\
		& \le \sum_{b=\ceil{\log_q(L+1)}}^{L+1} \PR{\exists A\in \F_q^{n \times b}\text{ s.t. }A \text{ is }\SC \rho {L+1}\text{ and }A\subseteq \Lambda\cC} \\
		& \le \sum_{b=\ceil{\log_q(L+1)}}^{L+1} \E{\inabset{A\in \F_q^{n \times b}\mid A \text{ is }\SC \rho {L+1}\text{ and }A\subseteq \Lambda\cC}}\enspace.
	\end{align*}
	%
	%By Remark \ref{rem:SpanClusteredRowSymmetric}
	We can write 
	$$\inset{A\in \F_q^{n\times b}\mid A \text{ is }\SC \rho {L+1}} = \bigcup_{\tau\in T_b} \cM_{n,\tau}$$ where $T_b$ is a set of $n$-feasible distributions over $\F_q^b$.  Therefore, by Lemma \ref{lem:tauInLargeDistance} and our assumption that $\eta \le q^{-{L+1}} \le q^{-b}$, the probability that $\Lambda\cC\text{ is not }\LD \rho L$ is at most
	\begin{align*} 
		&
		\sum_{b=\ceil{\log_q(L+1)}}^{L+1} \sum_{\tau\in T_{b}} \E{\inabset{A\in \cM_{n,\tau}\mid A\subseteq \Lambda\cC}} \\ 
		& ~ \le \sum_{b=\ceil{\log_q(L+1)}}^{L+1} \sum_{\tau\in T_{b}} \expOver q{n\cdot\inparen{ H_q(\tau)-(1-R)b+\log_q(1+\eta q^b)+\log_q 2}} \\
		& ~ \le\sum_{b=\ceil{\log_q(L+1)}}^{L+1} \sum_{\tau\in T_{b}} \expOver q{n\cdot\inparen{ H_q(\tau)-(1-R)b+2}}\enspace. 
	\end{align*}
	By Corollary \ref{cor:GHKEntropy}, each term of the inner sum is at most $q^{-n}$. Therefore, by Fact \ref{fact:NumberOfTypes},
	\begin{align*}
		\PR{\Lambda\cC\text{ is not }\LD \rho L} & \le \sum_{b=\ceil{\log_q(L+1)}}^{L+1} \sum_{\tau\in T_{b}} q^{-n}\\
		& \le q^{-n}\sum_{b=1}^{L+1} (n+1)^{q^b} \le q^{-n}(L+1) (n+1)^{q^{L+1}}\enspace , \end{align*}
	and the theorem follows due to our assumption that $\frac{n}{\log_q n}\ge \omega\inparen{q^{L+1}}$.
\end{proof}

\begin{remark}[On the conditions in Theorem \ref{thm:MainReduceToGHK}, and comparison to Theorems \ref{thm:MainLowBiasProperties} and \ref{thm:MainLargeDistanceProperties}]\label{rem:ReasonForConditions}
	In the above proof of Theorem \ref{thm:MainReduceToGHK}, as in the proofs of Theorems \ref{thm:MainLowBiasProperties} and \ref{thm:MainLargeDistanceProperties}, the core of the proof is obtaining an upper bound on terms of the form $\E{\inabset{A\in \cM_{n,\tau}\mid A\subseteq \Lambda \cC}}$ for certain distributions $\tau$, where $\cC$ is a random puncturing of a mother code $\cD$. 
	
	When our assumption about $\cD$ is that of near-optimal distance, we bound this expectation via Lemma \ref{lem:tauInLargeDistance}, which includes a bothersome $\log_q 2$ term. This term needs to be bounded from above by $a\eps$. One way to overcome this term is to take $q$ large enough to make $\log_q2$ negligibly small, as we do in Theorem \ref{thm:MainLargeDistanceProperties}. 
	
	In Theorem \ref{thm:MainReduceToGHK} we handle this problem differently. Appealing to \cite{GHK} (Corollary \ref{cor:GHKEntropy}), we obtain a stronger upper bound on the expected number of $\tau$-distributed matrices in $\cC$, circumventing the need to take $q$ to be large. Corollary \ref{cor:GHKEntropy} provides us with a ``slack'' that dominates the $\log_q 2$ term whenever $a$ is small, whereas for large $a$ the $a\eps$ upper bound is not too restrictive. However, this stronger upper bound is only valid when $\tau$ is a ``bad type'' for list-decoding, meaning that Theorem \ref{thm:MainReduceToGHK} only applies to the property of list-decodability, rather than to a general family of properties as Theorems \ref{thm:MainLowBiasProperties} and \ref{thm:MainLargeDistanceProperties} do.
\end{remark}

\section{Derandomization of RLCs}\label{sec:Derandomization}
In this section we prove Theorem \ref{thm:DerandomizationAllProperties}, restated here.
\DerandomizationAllProperties*
\begin{proof}
	Fix a property $\cP\in \cK$. Fix $\eta = \frac{\eps b\ln 2}{2^b}$. Let $\cD\subseteq \F_2^m$ be an $\eta$-biased linear code of dimension $Rn$, where $m\le O(n\cdot \eta^{-c})$ for some universal $c \ge 2$. Explicit constructions of such a code $\cD$ are given in \cite{ABNNR92,Ta-Shma17}. We also assume that 
	\begin{equation}\label{eq:DerandomizationMIsLarge}
		m \ge \frac{n}{1-2^{-\frac \eps 2}}\enspace,
	\end{equation}
	noting that $\frac m n$ can be taken to be as large as desired.
	
	Sample a random increasing sequence of $n$ integers $1\le i_1< i_2<\dots< i_n\le m$ uniformly from among all such sequences. Note that such a sequence can be encoded by $$\log_2 \binom m n+O(1)\le n \inparen{\log_2\frac mn+O(1)}\le O\inparen{n\inparen{b+\log_2\frac 1\eps}}$$ random bits, whose decoding can be done in $\poly(m)$ time. 
	
	Let $\cC$ be the code defined by the random sequence $i_1,\dots, i_n$ via $\cC = \inset{\inparen{u_{i_1}\dotsc u_{i_n}}\mid u\in \cD}\subseteq \F_2^n$. Clearly, a generating matrix for $\cC$ can be obtained from that of $\cD$ in $\poly(m) = \poly(n)$ time.  Hence, to prove the theorem it suffices to show that $\cC$ satisfies Eq.\ \eqref{eq:DerandomizationALlPropertiesWant1}. 
	Let $\cC'\subseteq \F_2^n$ be a random $n$-puncturing of $\cD$. Let $T$ be the event that $\cC'$ satisfies $\cP$. Let $J$ denote the event that no coordinate of $\cD$ is sampled more than once for inclusion in $\cC'$. Note that $\PR{J} \ge \inparen{1-\frac nm}^n$. By Theorem \ref{thm:MainLowBiasProperties}, $\PR{T} \le 2^{-(\eps - o(1))n}$.
	Thus,
	$$\PR{T\mid J}\le \frac{\PR{T}}{\PR{J}}\le \expOver{2}{\inparen{-\eps-\log_2\inparen{1-\frac nm}+o(1)}n} \le 2^{-\Omega(\eps n)}\enspace,$$
	where the last inequality follows from Eq.\ \eqref{eq:DerandomizationMIsLarge}. 
	
	By row-symmetry, $\cP$ is invariant to coordinate permutations of $\cC$. Observe that a uniformly random coordinate permutation of $\cC$ yields a code distributed identically to the distribution of $\cC'$ conditioned on the event $J$. Therefore,
	\begin{equation}\label{eq:DerandomizationUpperBoundForOneProperty}
		\PR{\cC\text{ satisfies }\cP} = \PR{T\mid J} \le 2^{-\Omega(\eps n)}\enspace
	\end{equation}
	for every $\cP\in \cK$.

	It remains to show that Eq.\ \eqref{eq:DerandomizationUpperBoundForOneProperty} implies Eq.\ \eqref{eq:DerandomizationALlPropertiesWant1}. Let $\cK' = \inparen{\cP\in \cK\mid \inabs{\cT_{\cP}}=1}$ (recall Fact \ref{fact:PropertyDecomposition} for the definition of $\cT_{\cP}$). Observe that a necessary condition for the event in Eq.\ \eqref{eq:DerandomizationALlPropertiesWant1} is that $\cC$ satisfies some property in $\cK'$. Indeed, suppose that $\cC$ satisfies a property $\cP\in \cK$ and let $\tau \in \cT_\cP$ such that $\cC$ contains a matrix in $\cM_{n,\tau}$. Let $\cP'$ denote the $b$-local, row-symmetric and monotone-increasing property for which $\cT_{\cP'} = \{\tau\}$. Clearly, $\cC$ satisfies $\cP'$. 
	Since $\cP'$ implies $\cP$, we have $\TRLC(\cP') \ge \TRLC(\cP) \ge R^*$ and so $\cP' \in \cK'$.   Thus, to prove the theorem, it suffices to show that
	\begin{equation}\label{eq:DerandomizationALlPropertiesWant2}
		\PR{\cC\text{ satisfies some property }\cP'\in \cK'} \le 2^{-\Omega(\eps n)}\enspace.
	\end{equation}
	Now, by Fact \ref{fact:NumberOfTypes}, $\inabs{\cK'} \le (n+1)^{2^b} \le 2^{o(n)}$. Thus, Eq.\ \eqref{eq:DerandomizationALlPropertiesWant2} follows from Eq.\ \eqref{eq:DerandomizationUpperBoundForOneProperty} by a union bound on $\cK'$, noting that $\cK'\subseteq \cK$. 		
\end{proof}

\section{Random puncturings of low-bias codes achieve capacity versus memoryless additive noise}\label{sec:stochastic}

Here we prove Theorem \ref{thm:MainStochastic}. 
\MainStochastic*
We require the following lemma.
\begin{lemma}\label{lem:ProbInLowBiasCode}
	Let $C\subseteq\F_q^n$ be a random rate $R$ puncturing of an $\eta$-biased linear code. Fix $x \in \F_q^n\setminus \{0\}$. Then, 
	$$\PROver{\cC}{x\in \cC} \le q^{n\cdot\inparen{-1+R+(q-1)\eta}}\enspace.$$
\end{lemma}
\begin{proof}
	Denote the mother code by $\cD\subseteq \F_q^m$ and let $\varphi:\F_q^m\to \F_q^n$ be the puncturing map. By the union bound, 
	\begin{align*}
		\PR{x\in \cC} &\le \sum_{u\in \cD} \PR{\varphi(u) = x} = \sum_{u\in \cD\setminus \{0\}} \prod_{i=1}^n \PR{\varphi_i(u) = x_i} \\ 
		&= \sum_{u\in \cD\setminus \{0\}} \prod_{i=1}^n \frac {\inabset{j\in [m]\mid u_j = x_i}}m \\
		&\le \sum_{u\in \cD\setminus \{0\}} \prod_{i=1}^n \inparen{1-\frac{q-1}q(1-\eta)} & \text{(by Lemma \ref{lem:VectorBiasAndWeight})}\\
		&= \sum_{u\in \cD\setminus \{0\}} \inparen{\frac{1+(q-1)\eta}q}^n \\
		&\le q^{Rn}\inparen{\frac{1+(q-1)\eta}q}^n \le q^{n\cdot\inparen{-1+R+(q-1)\eta}} \ . {\hfill \qedhere}
	\end{align*}
\end{proof}

\begin{proof}[Proof of Theorem \ref{thm:MainStochastic}]
	Let $J$ denote the event that our code $\cC$ is MLDU-decodable with error probability at most $2q^{-c_\nu \eps^2 n}$  with regard to the $\nu$-memoryless additive noise channel (where $c_\nu >0$ shall be chosen later). In other words, $J$ means that  Eq.\ \eqref{eq:ProbOfCorrectDecoding} holds for all $x\in \cC$. By  linearity of $\cC$, a sufficient condition for the latter is that Eq.\ \eqref{eq:ProbOfCorrectDecoding} holds for $x=0$. By the definition of the MLDU, the codeword $0$ with noise vector $z$ is decoded correctly whenever $\nu(z) > \nu(z-x)$ for all $x\in \cC\setminus \{0\}$. Thus,
	\begin{equation}
		\text{The code $\cC$ satisfies $J$ }\quad\text{if}\quad\PROver{z\sim \nu^n}{\forall x\in \cC\setminus \{0\}~~\nu(z) > \nu(z-x)} \ge 1-2q^{-c_\nu\eps^2 n}\enspace. \label{eq:StochasticProof1}
	\end{equation}
	
	Let $z\sim \nu^n$ and let $M_z$ denote the event that $z$ belongs to an \deffont{$\frac \eps3$-typical-set}, namely, $q^{-H_q(\nu)n-\frac {\eps n}3}\le \nu(z)\le q^{-H_q(\nu)n+\frac {\eps n}3}$. It is well known (e.g., \cite{YangJui2012}) that 
	$$\PROver{z\sim \nu^n}{M_z} \ge 1 - q^{-c'_\nu\eps^2 n}\enspace,$$
	for some positive constant $c'_\nu$.
	
	Denote $E_z = \inset{x\in \F_q^n\setminus \{0\}\mid  \nu(z)\le \nu(z-x)}$. Since $\sum_{x\in \F_q^n} \nu(z-x) = \sum_{x\in \F_q^n}\nu(x) = 1$, we have $|E_z|\le \frac 1{\nu(z)}$. Now, for a fixed code $\cC$,
	\begin{align*}
		{\PROver{z\sim \nu^n}{\exists x\in \cC\setminus \{0\}~~\nu(z) \le \nu(z-x)}} &= {\PROver{z\sim \nu^n}{E_z\cap \cC \ne \emptyset}}
		\\ &\le {\PROver{z\sim \nu^n}{E_z\cap \cC \ne \emptyset\mid M_z} + \PROver{z\sim \nu^n}{\overline {M_z}}} \\ &\le {\PROver{z\sim \nu^n}{E_z\cap \cC \ne \emptyset\mid M_z} + q^{-c'_\nu\eps^2 n}}\enspace.\numberthis\label{eq:StochasticProof2}
	\end{align*}
	Let $\eta = \frac{\eps}{3(q-1)}$. By Lemma \ref{lem:ProbInLowBiasCode}, for any $z$ such that $M_z$ holds, we have
	\begin{align*}
		\PROver{\cC}{E_z\cap \cC \ne \emptyset} &\le \sum_{x\in E_z}\PROver{\cC}{x\in \cC} \\&\le \sum_{x\in E_z} q^{n\cdot\inparen{-1+R+(q-1)\eta}} \le |E_z|\cdot q^{n\cdot\inparen{-1+R+(q-1)\eta}} \\&\le \frac{1}{\nu(z)}\cdot q^{n\cdot\inparen{-1+R+(q-1)\eta}} \\ &\le q^{n\cdot\inparen{-1+R+(q-1)\eta+H_q(\nu)+\frac \eps 3}}\\ &\le q^{-\frac{\eps n}3}\enspace.\numberthis\label{eq:StochasticProof3}
	\end{align*}
	
	Now,
	\begin{align*}
		&\PROver{\cC}{\cC \text{ does not satisfy }J} \\&\le \ \PROver{\cC}{\PROver{z\sim \nu^n}{\forall x\in \cC\setminus \{0\}~~\nu(z) > \nu(z-x)}< 1-2q^{-c_\nu\eps^2 n}}  & (\text{by Eq.\ \eqref{eq:StochasticProof1}} 
		\\&= \ \PROver{\cC}{\PROver{z\sim \nu^n}{\exists x\in \cC\setminus \{0\}~~\nu(z) \le \nu(z-x)} \ge 2q^{-c_\nu\eps^2 n}}\\
		&\le \  \PROver{\cC}{\PROver{z\sim \nu^n}{E_z\cap \cC\ne \emptyset\mid M_z}  \ge 2q^{-c_\nu\eps^2 n} - q^{-c'_\nu\eps^2 n}} & (\text{by Eq.\ \eqref{eq:StochasticProof2}})\\ 
		&\le \ \frac{\Eover{\cC}{\PROver{z\sim \nu^n}{E_z\cap \cC\ne \emptyset\mid M_z}}}{2q^{-c_\nu\eps^2 n} - q^{-c'_\nu\eps^2 n}} & (\text{by Markov's inequality})\\
		&\le \ \frac{q^{-\frac{\eps n}3}}{2q^{-c_\nu\eps^2 n} - q^{-c'_\nu\eps^2 n}}\enspace. &(\text{by Eq.\ \eqref{eq:StochasticProof3}})
	\end{align*}
	Taking $c_\nu = \min\inset{c_{\nu'}, \frac 14}$ makes the right-hand side $q^{-\Omega_\nu(\eps n)}$, finishing the proof.
\end{proof}

%%% AUTHOR: optional acknowledgments here
\section*{Acknowledgments}
The authors thank Jo\~ao Ribeiro for pointing out an improvement to Theorem \ref{thm:DerandomizationAllProperties}, and to Zeyu Guo for finding an error in a previous version of this paper.

%%% AUTHOR:
%%% Bibliography goes here. Note that the arXiv cannot process bibtex
%%% or biber bibliographies.  Example of acceptable bibliograpy format:
\bibliographystyle{amsplain}

\newcommand{\etalchar}[1]{$^{#1}$}

\appendix
\section {A direct characterization of span-clustered distribution}\label{appendix:SpanClusteredTypes}

In Definition \ref{def:SpanClusteredType}, we defined a span-clustered distribution as the empirical distribution of a span-clustered matrix. Here, we give an alternative definition that does not rely on the notion of a span-clustered matrix.
\begin{lemma}
    An $n$-feasible full-rank distribution $\tau$ over $\F_q^{b}$ is $\SC \rho{L+1}$ if and only if there exists a matrix $B \in \F_q^{b\times (L+1)}$ with distinct rows and an $n$-feasible distribution $\sigma$ over $\F_q^{b}\times \F_q$, such that a random pair $(x,z)$ sampled from $\sigma$ ($x\in \F_q^b$, $z\in \F_q$) satisfies the following.
    \begin{enumerate}
        \item The distribution of $x$ is $\tau$.
        \item For each $1\le i\le L+1$ it holds that $\PR{(x B)_i \ne z} \le \rho$. Here, we think of $x$ as a row vector.
    \end{enumerate}
    
\end{lemma}	
\begin{proof}
    \sloppy
    Suppose that $\tau$ is span-clustered according to Definition \ref{def:SpanClusteredType}. Hence, there exists a $\SC \rho{L+1}$ matrix $A\in \F_q^{n\times b}$ with $\tau = \distr A$. Clearly, $\tau $ is $n$-feasible and of full rank, the latter since $\rank A = b$. Since $A$ is $\SC \rho{L+1}$, there is a matrix $B\in \F_q^{b\times (L+1)}$ with distinct columns, and a vector $y\in \F_q^n$ such that $\weight((AB)_i - y) \le \rho$ for all $1\le i\le L+1$. 
    
    We take $\sigma$ to be the distribution, over $\F_q^{b}\times \F_q$, of the random pair $(A_j,y_j)$, where $j$ is sampled uniformly at random from $\{1,\dots, n\}$. Clearly, $\sigma$ is $n$-feasible, and $A_i$ is distributed $\tau$. Finally, for $1\le i\le L+1$,
    $$\PROver{(x,z)\sim \sigma}{(xB)_i\ne z} = \PROver{j\sim \uniform([n])}{\inparen{A_jB}_i \ne y_j} = \weight((AB)_i-y) \le \rho\enspace,$$
    showing that $\sigma$ satisfies the requirements of the lemma.
    
    In the other direction, suppose that there exists a distribution $\sigma$ as in the lemma statement, and let $A\in \F_q^{n\times b}$ such that $\distr A = \tau$. A straightforward reversal of the above argument shows that $A$ is $\SC{\rho}{L+1}$, and thus $\tau$ is $\SC{\rho}{L+1}$ by Definition \ref{def:SpanClusteredType}.
\end{proof}

%% AUTHOR: You can generate such a bibliography from a .bib file by 
%% running pdflatex/bibtex/pdflatex/pdflatex and then pasting the .bbl file
%% between \begin{thebibliography} and \end{bibliography}

%%% AUTHOR: Include a short description of each author following the
%%% structure below. Use the same short tags used previously.  
%%% Use \imageat{} and \imagedot{} instead of "@" and "." in
%%% email addresses-this replaces the symbols with graphics to avoid 
%%% e-mail address harvesting from the .pdf file
\begin{dajauthors}
\begin{authorinfo}[vg]
  Venkatesan Guruswami\\
  University of California, Brekeley\\
  Berkeley, California, USA\\
  venkatg\imageat{}berkeley\imagedot{}edu \\
  \url{https://people.eecs.berkeley.edu/~venkatg/}
\end{authorinfo}
\begin{authorinfo}[jm]
  Jonathana Mosheiff\\
  Ben-Gurion University\\
  Be'er Sheva, Israel\\
  mosheiff\imageat{}bgu\imagedot{}ac\imagedot{}il \\
  \url{https://cs.bgu.ac.il/~mosheiff}
\end{authorinfo}
\end{dajauthors}

\end{document}